\documentclass[12pt,a4paper]{amsart}

\usepackage[utf8]{inputenc}
\usepackage{enumerate,commath}
\usepackage[page,toc]{appendix}
\usepackage[dvipsnames]{xcolor}
\usepackage[colorlinks,linkcolor=BrickRed,citecolor=OliveGreen,urlcolor=black,hypertexnames=true]{hyperref}
\usepackage{tikz}
\usetikzlibrary{calc,decorations.pathmorphing,shapes}
\usepackage[T1]{fontenc}

\usepackage{csquotes}
\usepackage[giveninits=true,isbn=false,url=false,doi=false,eprint=false,style=alphabetic,minbibnames=3,maxbibnames=3,minalphanames=3,maxalphanames=3]{biblatex}
\AtBeginBibliography{\small}

\DeclareMathOperator{\supp}{supp}

\DeclareMathOperator{\spa}{span}
\DeclareMathOperator{\ran}{ran}

\DeclareMathOperator{\tr}{tr}
\DeclareMathOperator{\rk}{rk}
\DeclareMathOperator{\opm}{M}

\newcommand{\mat}{{\bf mat}}

\newcommand{\N}{\mathbb{N}}
\newcommand{\Z}{\mathbb{Z}}
\newcommand{\R}{\mathbb{R}}
\newcommand{\C}{\mathbb{C}}

\newcommand{\cA}{\mathcal{A}}
\newcommand{\cB}{\mathcal{B}}
\newcommand{\cC}{\mathcal{C}}
\newcommand{\cF}{\mathcal{F}}
\newcommand{\cH}{\mathcal{H}}

\newcommand{\cK}{\mathcal{K}}

\newcommand{\cN}{\mathcal{N}}

\newcommand{\ti}{{\rm t}}
\newcommand{\hA}{\hat{A}}
\newcommand{\hB}{\hat{B}}
\newcommand{\hC}{\hat{C}}

\newcommand{\uA}{\text{A}}

\newcommand{\povm}{{\bf povm}}

\newcommand{\bra}[1]{\mathinner{\langle #1|}}
\newcommand{\ket}[1]{\mathinner{|#1\rangle}}
\newcommand{\braket}[2]{\mathinner{\langle #1|#2\rangle}}
\newcommand{\dyad}[1]{| #1\rangle\!\langle #1|}
\newcommand{\ketbraq}[1]{\dyad{#1}}
\newcommand{\ketbra}[2]{| #1\rangle\!\langle #2|}

\newcommand*{\mtxc}[1]{\opm_{#1}(\C)}
\newcommand*{\mtxr}[1]{\opm_{#1}(\R)}

\newtheorem{thm}{Theorem}[section]
\newtheorem{lem}[thm]{Lemma}
\newtheorem{prop}[thm]{Proposition}
\theoremstyle{remark}
\newtheorem{exa}[thm]{Example}
\newtheorem{rem}[thm]{Remark}
\numberwithin{equation}{section}

\title{Maximal device-independent randomness in every dimension}

\author[Farkas]{M\'at\'e Farkas}
\address{Department of Mathematics, University of York, United Kingdom}
\email{mate.farkas@york.ac.uk}

\author[Vol\v{c}i\v{c}]{Jurij Vol\v{c}i\v{c}}
\address{Department of Mathematics, Drexel University, Pennsylvania}
\email{jurij.volcic@drexel.edu}

\author[Storgaard]{Sigurd A.~L.~Storgaard}
\address{Department of Mathematical Sciences, University of Copenhagen, Denmark}
\email{sals@math.ku.dk}

\author[Chen]{Ranyiliu Chen}
\address{Department of Mathematical Sciences, University of Copenhagen, Denmark}
\email{rc@math.ku.dk}

\author[Man\v{c}inska]{Laura Man\v{c}inska}
\address{Department of Mathematical Sciences, University of Copenhagen, Denmark}
\email{mancinska@math.ku.dk}

\thanks{RC, LM, SS were in part supported by Villum Fonden via Villum Young Investigator grant (No 37532). LM and SS additionally acknowledge support from ERC grant (QInteract, Grant Agreement No 101078107).
JV was supported by the NSF grant DMS-2348720}

\date{\today}

\keywords{Randomness, conditional quantum entropy, device-independent certification, Bell inequalities, informationally complete positive operator-valued measures}

\begin{document}

\begin{abstract}
Random numbers are used in a wide range of sciences. In many applications, generating unpredictable \textit{private} random numbers is indispensable. Device-independent quantum random number generation is a framework that makes use of the intrinsic randomness of quantum processes to generate numbers that are fundamentally unpredictable according to our current understanding of physics. While device-independent quantum random number generation is an exceptional theoretical feat, the difficulty of controlling quantum systems makes it challenging to carry out in practice. It is therefore desirable to harness the full power of the quantum degrees of freedom (the dimension) that one can control. 
It is known that no more than $2 \log(d)$ bits of private device-independent randomness can be extracted from a quantum system of local dimension  $d$. In this paper we demonstrate that this bound can be achieved for all dimensions $d$ by providing a family of explicit protocols.
In order to obtain our result, we develop new certification techniques that can be of wider interest in device-independent applications for scenarios in which complete certification (`self-testing') is impossible or impractical.

\end{abstract}

\maketitle

\begingroup
\setlength{\parskip}{0pt}
\tableofcontents
\endgroup

\section{Introduction}\label{sec:intro}

Randomness is an essential resource in many disciplines of science. It is useful for generating samples for simulation \cite{KGV83} or training data, and various computational models rely on randomized algorithms \cite{arora2009}. For some of these applications, the \textit{privacy} of the random numbers is not essential. That is, it is not a problem if third parties have access to the random numbers. In fact, in some cases deterministic \textit{pseudo-random} numbers are advantageous for reproducibility. In other applications, however, truly unpredictable \textit{private random numbers} are indispensable. This means that
no third party should be able to a priori guess the random numbers. A prominent application is cryptography, where random numbers are widely used in encryption and decryption schemes \cite{Sha48}. It is crucial for the security of such cryptographic protocols that these random numbers are private, and only
the encoder and the decoder have access to them. In particular, no
potential eavesdropper should be able to predict these numbers.

Any random number is ultimately generated by some physical process---a roll of a die, an output of a computer algorithm, atmospheric noise, etc. If this process is described by classical
physics, then the generated number is technically pseudo-random: perfect knowledge of the initial conditions of the system generating the
random numbers makes it possible to perfectly predict these numbers, due to the deterministic nature of classical physical laws. Therefore, to generate truly private random numbers, the underlying physical process must be quantum: the fundamentally probabilistic nature of quantum theory makes it impossible to predict the outcome of certain quantum mechanical experiments. Therefore, quantum measurements have the potential to generate private randomness.

Importantly, unpredictability for the user does not necessarily imply private randomness, even if quantum mechanical processes were used to generate the random numbers. Consider the case of generating randomness by measuring the quantum state $\frac{1}{\sqrt{2}}(\ket{0} + \ket{1})$ in the basis $\{ \ket{0}, \ket{1} \}$. The outcome of this measurement is a perfectly random bit. However, the privacy of this bit cannot be guaranteed unless the state and the measurement are characterized and trusted. Indeed, these measurement statistics are also compatible with measuring one half of a bipartite state $\frac{1}{\sqrt{2}}(\ket{0} \otimes \ket{0} + \ket{1} \otimes \ket{1})$ in the basis $\{ \ket{0}, \ket{1} \}$. An eavesdropper having access to the other half can then perfectly predict the outcome of this measurement by also measuring in the basis $\{ \ket{0}, \ket{1} \}$, and thus the randomness is not private.

The framework that allows for certifying private randomness only from observed data (without a priori characterization of states and measurements) is called \textit{device-independent quantum random number generation} (DIQRNG) \cite{Colbeck2007,Pironio_2010,VV12}. This framework requires the use of a bipartite quantum system, and the private randomness of measurements performed on one half of the system can be certified based on the full bipartite measurement statistics. Crucially, this certification is based only on the correctness of quantum theory, and on the assumption that the two laboratories measuring the two systems do not leak information (a necessary assumption in any private randomness certification protocol). Importantly, DIQRNG replaces unverifiable computational assumptions (often used in classical cryptography \cite{RSA78,Regev09}, and in some quantum protocols \cite{Metger_2021,mazzoncini2023}) with information-theoretic security.

Various DIQRNG protocols have been proposed and experimentally demonstrated \cite{LXW+21,Shalm2021,Liu2021}, and less secure, device-\textit{dependent}, variants are also available commercially \cite{IDQ,Toshiba}. One of the main challenges of true DIQRNG is its resource-intensity. In order to generate private random numbers, high-quality entangled quantum systems and quantum measurement devices need to be manufactured and operated reliably. We address the fundamental question of how much device-independent randomness can be generated using fixed quantum resources. Namely, we fix the controllable degrees of freedom---the \textit{dimension}---of an entangled quantum state. Scaling up fully controllable dimensions is a major practical challenge in quantum technologies \cite{Arute2019,Madsen2022,Acharya2023}, and therefore taking full advantage of the available resources is crucial.

It is known that if a bipartite quantum state is locally $d$-dimensional, a fundamental upper bound on the certifiable private randomness is $2 \log(d)$ bits. Apart from the case of $d=2$ \cite{APVW16}, however, it was previously unknown whether this fundamental limit can be achieved, and if so, what protocol should be used. In this work, we fully solve this problem. We show that the fundamental limit of $2 \log(d)$ bits of private (device-independent) randomness can be extracted from locally $d$-dimensional systems in every dimension $d$. Moreover, our proof is constructive, providing an explicit protocol that certifies the maximal randomness. The protocol uses a locally $d$-dimensional maximally entangled state, but there is significant freedom in the choice of measurement. Our proof techniques are similar to \textit{self-testing}, which is a powerful certification tool in quantum information theory \cite{Supic2020selftestingof}. While self-testing allows for essentially uniquely identifying quantum states and measurements from observed experimental statistics, our certification tools allow for more freedom, certifying only the properties essential for randomness certification. We believe that our framework opens up new possibilities for practical device-independent certification methods in scenarios where self-testing is too demanding or unnecessary.

\section{Preliminaries}

The basic setup of DIQRNG is a so-called bipartite \textit{Bell scenario} \cite{BCP+14}. Two experimenters, commonly referred to as Alice and Bob, share a bipartite physical system and perform various measurements on their respective parts. They repeat the experiment many times in order to be able to estimate the outcome probabilities of the measurements. All conclusions of the experiment are drawn from the measurement statistics, and from the assumptions that quantum theory is correct and Alice's and Bob's laboratories do not leak any information. Crucially, Alice and Bob do not need any prior knowledge of the physical system or their measurement devices.

In quantum theory, physical systems are associated with a complex Hilbert space $\cH$, which we assume to be finite-dimensional in this work. We will often refer to the set of matrices (linear operators) on a $d$-dimensional complex Hilbert space as $M_d(\C)$. A bipartite system is associated with a tensor product of two Hilbert spaces, $\cH_A \otimes \cH_B$, and the state of the system is described by a positive semidefinite operator $\rho$ on this Hilbert space with unit trace. Measurements on a Hilbert space are described by positive operator-valued measures (POVMs), which in our case correspond to a tuple of positive semidefinite operators that add up to the identity $I$. In a Bell scenario, Alice has various measurement settings labeled by $x$, and the possible outcomes of her measurements are labeled by $a$ (the number of possible outcomes may depend on $x$). Similarly, Bob's settings are labeled by $y$ and his outcomes by $b$. An experiment is characterized by the \textit{correlation}, that is, the probabilities of Alice observing outcome $a$ and Bob $b$, upon choosing measurement settings $x$ and $y$. According to quantum theory, these probabilities are given by
\begin{align}\label{eq:correlation}
p(a,b|x,y) = \tr[ \rho (A^x_a \otimes B^y_b) ],
\end{align}
where $\rho$ is a quantum state on $\cH_A \otimes \cH_B$, $(A^x_a)_a$ is a POVM on $\cH_A$ for all $x$ and $(B^y_b)_b$ is a POVM on $\cH_B$ for all $y$.

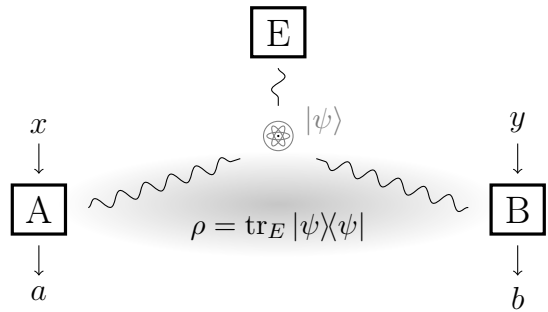
\begin{figure}
    \centering

\begin{tikzpicture}

\shadedraw[inner color=gray!45, outer color=white,draw=white]
(0,-.6) ellipse (2.9cm and .7cm);

\draw[black, very thick] (-3.5,-1) rectangle (-2.8,-.35)
node[below,xshift=-.35cm,yshift=0cm] {{\large A}};
\draw[->] (-3.15 , .2) -- (-3.15 , -.2) node[above,yshift=.4cm] {$x$};
\draw[->] (-3.15 , -1.15) -- (-3.15 , -1.55) node[below,xshift=0cm] {$a$};

\draw[black, very thick] (3.5,-1) rectangle (2.8,-.35)
node[below,xshift=.35cm,yshift=0cm] {{\large B}};
\draw[->] (3.15 , .2) -- (3.15 , -.2) node[above,yshift=.4cm] {$y$};
\draw[->] (3.15 , -1.15) -- (3.15 , -1.55) node[below,xshift=0cm] {$b$};

\draw[black, very thick] (-.35,1.35) rectangle (.35,2)
node[below,xshift=-.35cm,yshift=0cm] {{\large E}};
\def\science#1#2 {
\node[circle, draw, minimum size=2.7mm,color=gray] (#2) at #1 {};
\foreach \ang in {0,120,240}
\draw[rotate around={\ang:#1},color=gray] #1 ellipse (1.5mm and .5mm);
\fill #1 circle (0.15mm);
}
\science{(0,.3)}{C1};
\path[draw,-,decorate,decoration={snake}] (-2.5,-.65) -- (-.5,0);
\path[draw,-,decorate,decoration={snake}] (0,1.2) -- (0,.7)
node[right, xshift=.2cm,yshift=-.2cm] {\textcolor{gray}{\small{$\ket{\psi}$}}};

\path[draw,-,decorate,decoration={snake}] (2.5,-.65) -- (.5,0);
\node at (0,-0.9) {\small{$\rho= \tr_E\dyad{\psi}$}};
\end{tikzpicture}

\caption{Bell scenario including an eavesdropper.}
\label{fig:setup_main}

\end{figure}

The aim of DIQRNG is to lower bound the randomness of certain measurement outcomes from the point of view of a potential eavesdropper, Eve, who holds part of a tripartite purification of $\rho$ (see Figure \ref{fig:setup_main} for a schematic representation). Moreover, this bound must be solely based on the observed correlation.
In this work, we are interested in the asymptotic rate of randomness of the outcome of a single setting $x$ of Alice. This rate quantifies the randomness of a single measurement outcome of the measurement $x$, in the limit of Alice and Bob performing the experiment infinitely many times. Notice that while this setup is formulated under the assumption that the quantum state and measurements behave the same way in every experimental round (i.i.d.~assumption), this assumption can be lifted in the asymptotic limit due to the entropy accumulation theorem \cite{Arnon-Friedman2018,Dupuis2020}. It is then well-known \cite{TCR09} that the asymptotic rate of randomness is lower bounded by the conditional von Neumann entropy $H(\uA|E)_{\rho_{\uA E}}$ of the classical-quantum state
\begin{align}\label{eq:cqstate}
\rho_{\uA E} = \sum_a \ketbraq{a}_\uA \otimes \tr_{AB}[ \ketbraq{\psi}( A^x_a \otimes I_B \otimes I_E )],
\end{align}
where $(A^x_a)_a$ is the POVM describing the measurement for setting $x$ and $\ket{\psi}$ is a purification of the state $\rho$ from Eq.~\eqref{eq:correlation}. Note that we denote the classical register (post-measurement) of Alice by an upright $\uA$ and the quantum register (pre-measurement) by an italic $A$. Since in a device-independent setting the only thing we assume is the observed correlation, $H(\uA|E)_{\rho_{\uA E}}$ needs to be minimized over all possible physical realizations compatible with the observed correlation $p(a,b|x,y)$ or with some function $f[ p(a,b|x,y) ]$ of it. That is, the minimization is taken over all possible Hilbert spaces $\cH_A, \cH_B$ and $\cH_E$ and all states $\ket{\psi} \in \cH_A \otimes \cH_B \otimes \cH_E$ and POVMs $(A^x_a)_a$ on $\cH_A$ and $(B^y_b)_b$ on $\cH_B$ such that
\begin{align*}
    f[ p(a,b|x,y) ] = f[ \bra{\psi} (A^x_a \otimes B^y_b \otimes I_E) \ket{\psi} ].
\end{align*}
In general, this minimization is an extremely difficult task. In certain small and/or symmetric scenarios, analytic results exist, demonstrating that DIQRNG is possible in principle \cite{Colbeck2007,Pironio_2010,VV12,Arnon-Friedman2018}. General methods for bounding the rate of randomness usually rely on numerics \cite{Bancal_2014,Brown2021,Brown2024}, and scale rather badly with the number of measurement settings and outcomes.

In this work, we analytically carry out the above minimization for a class of correlations. In particular, we are interested in bounding the device-independent randomness from certain correlations that arise from measuring a quantum state that is locally $d$-dimensional, that is, $\cH_A = \cH_B = \C^d$ (note that while we are interested in correlations generated by locally $d$-dimensional states, we certify the randomness device-independently, that is, solely from the observed correlation). It is known that in this case the maximum possible certifiable device-independent randomness is upper bounded by $2 \log(d)$ (see \cite{APVW16} for a proof for the min-entropy), and we include a proof of this fact in Appendix \ref{app:bound} for completeness (for the von Neumann entropy).

We prove that this fundamental bound can be achieved in every dimension $d$ and we provide explicit protocols achieving this bound. This result significantly extends on the current state-of-the-art: to the best of our knowledge, up to now this result was only known for $d=2$ \cite{APVW16}. Furthermore, the techniques we develop have the potential to be widely useful across device-independent quantum information processing.

\section{Setup and results}

Our setup builds on the Bell scenario studied in \cite{tavakoli21}. The authors there propose a DIQRNG protocol and numerically prove that the protocol reaches close to $2$ bits of randomness with locally 2-dimensional systems. For locally 3-dimensional systems, the authors can numerically certify approximately $3.03 $ bits of randomness, falling somewhat short of the fundamental bound $2 \log(3) \approx 3.17$. Extending the numerical verification to larger dimensions is computationally extremely demanding. Moreover, extending the methods of \cite{tavakoli21} to arbitrary dimensions relies on the conjectured existence of symmetric informationally complete (SIC) POVMs in every dimension \cite{Zauner}.

In this work, we develop new certification techniques and design explicit protocols that certify $2 \log(d)$ bits of randomness for every $d$, overcoming previous limitations. Our results are completely analytical and use the family of measurements that we introduce in this work, called \textit{balanced informationally complete} (BIC) POVMs.

Our Bell scenario is parametrized by an integer $d \ge 2$ and a $d^2 \times d^2$ matrix $S = (s_{jk})_{j,k}$ that corresponds to a BIC-POVM. Formally, a BIC-POVM in dimension $d$ is a POVM of the form $\left(\frac1d P_j\right)_{j=1}^{d^2}$, where $P_j$ are rank-1 projections that form a basis for $M_d(\C)$. BIC-POVMs can be shown to exist in every dimension $d$ (see Appendix \ref{sec:icpovm}), and in our setup we let $s_{jk} = \tr(P_j P_k)$ where $\left(\frac1d P_j\right)_{j=1}^{d^2}$ is a BIC-POVM. In the Bell scenario, Alice has $d^2(d^2-1)/2 + 1$ measurement settings. The first $d^2(d^2-1)/2$ settings are labeled by pairs $jk$ from the set $\text{Pairs}(d^2):=\{(j,k) \in [d^2] \times [d^2] \colon j<k\}$, and these measurements have three outcomes each labeled by $a \in \{1,2,\perp\}$. The last measurement setting of Alice is labeled by $\povm$ and this has $d^2$ outcomes labeled by $a \in [d^2]$. This is the measurement that Alice uses to generate randomness. On the other side, Bob has $d^2$ settings with two outcomes each, labeled by $b\in\{1,\perp\}$.

We introduce a \textit{Bell function} (a linear function of correlations, equivalently referred to as a \textit{Bell inequality}) with the aim that reaching its maximal value certifies the maximally entangled two-qudit state $\ket{\varphi_d} := \frac{1}{\sqrt{d}} \sum_{j=1}^d \ket{j} \otimes \ket{j}$ and a BIC-POVM $A^{\povm}_j = \frac1d \ketbraq{e_j} $ such that $\abs{\braket{e_j}{e_k}}^2 = s_{jk}$ for all $j,k$. The function reads
\begin{equation}
\begin{split}
&  2\sum_{j<k}\sqrt{1-s_{jk}} \Big[ p(1,1|jk,j) + p(2,1|jk,k) - p(1,1|jk,k) - p(2,1|jk,j) \Big]  \\
& - \sum_{j<k}(1-s_{jk}) \big[ p_A(1|jk) + p_A(2|jk) \big]  - d(d-2) \sum_{j=1}^{d^2} p_B(1|j) - \sum_{j=1}^{d^2} p(j,\perp|\povm,j),
\end{split}
\end{equation}
where $\sum_{j<k}$ is short-hand for $\sum_{j=1}^{d^2-1} \sum_{k=j+1}^{d^2}$, and $p_A$ and $p_B$ are Alice's and Bob's marginal probabilities, respectively.

Every Bell function has a corresponding \textit{Bell operator}, which is an operator-valued function of POVM elements $A^x_a$ and $B^y_b$. In our case, the Bell operator reads
\begin{equation} \label{e:bell_operator}
\begin{split}
W_d :=\, & 2\sum_{j<k} \sqrt{1-s_{jk}}(A^{jk}_1 - A^{jk}_2) \otimes (B_j - B_k)  -  \sum_{j<k} (1-s_{jk}) (A^{jk}_1 + A^{jk}_2)\otimes I \\  
&-  d(d-2)\sum_{j=1}^{d^2} I\otimes B_j - \sum_{j=1}^{d^2} A^\povm_j \otimes (I - B_j),    
\end{split}
\end{equation}
where we used the notation $B_j := B^j_1$ (and therefore $B^j_\perp = I - B_j$). Note that in our notation we suppressed the dependence on the $S$ matrix, as well as the fact $W_d$ is a function of $A^x_a$ and $B^y_b$. For a given set of POVMs with elements $A^x_a$ and $B^y_b$ and a given bipartite state $\rho$, the value of the Bell inequality is given by $\tr(W_d \rho)$.

One common way to show that some $\beta \in \mathbb{R}$ is an upper bound on the value of a Bell inequality is through a sum-of-squares decomposition of the \textit{shifted} Bell operator $\beta I - W_d$. Specifically, if we can write
\begin{equation}\label{e:SOS_general}
\beta I - W_d = \sum_i P_i^* P_i,
\end{equation}
for some operator-valued functions $P_i$ of $A^x_a$ and $B^y_b$ (assuming that these form valid POVMs), then we have that for every set of POVMs with elements $A^x_a$ and $B^y_b$ and for every quantum state $\rho$
\begin{equation}\label{e:Bell_bound}
\tr(W_d \rho) \le \beta,
\end{equation}
which is equivalent to saying that $\beta$ is an upper bound on the value of the Bell inequality.
The inequality in Eq.~\eqref{e:Bell_bound} comes from the fact that $P_i^*P_i$ is positive semidefinite for every $A^x_a$ and $B^y_b$ and that $\tr(\rho) = 1$ for every valid quantum state.

If Eq.~\eqref{e:Bell_bound} can be saturated then $\beta$ is a tight bound. In this case, for every set of POVMs with elements $A^x_a$ and $B^y_b$ and every state $\rho$ such that $\tr(W_d \rho) = \beta$ it must hold that $P_i \rho = 0$ for all $i$. In some cases, these relations make it possible to essentially uniquely (up to local isometries) identify the quantum state and the POVMs that give rise to the maximal Bell violation. This is called \textit{self-testing}, and it is a powerful tool in device-independent quantum information processing \cite{Supic2020selftestingof}.

Importantly, in our DIQRNG protocol we would like to use a non-projective POVM (a BIC-POVM), and it is known that non-projective measurements cannot be self-tested due to Naimark's dilation theorem \cite{baptista23}. To certify our setup, we therefore develop new, weaker forms of self-testing. Based on the maximal value of our Bell inequalities we characterize the POVMs and the quantum state sufficiently so that we are able to bound the conditional von Neumann entropy of any state of the form in Eq.~\eqref{eq:cqstate} compatible with the maximal value. Similar weaker forms of self-testing have been studied recently in the context of device-independent quantum information processing \cite{kan20,Far24}. Our techniques differ from these by certifying measurements \emph{compressed onto the local support} of the state (see below), and by using the representation theory of measurement algebras. Crucially, we do not make any a priori assumptions on the state and measurements \cite{baptista23}. Our methods appear to be highly promising for further device-independent applications beyond this work, in scenarios in which a complete self-testing statement is impossible or impractical.

Crucial in our analysis are the local supports of the state $\rho \in \mtxc{d_A} \otimes \mtxc{d_B}$. We define $\supp_A \rho$ as the range of $\tr_B \rho \in \mtxc{d_A}$, and $\supp_B \rho$ in an analogous way. We then define \emph{compressions} onto the local supports. On Alice's Hilbert space, the compression of an operator $X$ is defined as $\hat X := U X U^*$, where $U: \supp_A \rho \to \C^{d_A}$ is the inclusion, that is, $U^*U$ is the identity on $\supp_A$ and $UU^*$ is the projection on $\C^{d_A}$ with range $\supp_A \rho$. Compressions on Bob's Hilbert space are defined analogously, also denoted by $\hat X$. The intuition behind considering compressed operators is that one cannot certify anything about the measurements outside the support of the state.

The following proposition informally summarizes our certification results based on the maximal value of our Bell inequalities.
\begin{prop} \label{prop_Bell_certificate}
The maximum quantum value of the Bell function $W_d$ in Eq.~\eqref{e:bell_operator} is~$d^2$ for every $S$ induced by a BIC-POVM. Furthermore, if $d^2$ is reached using the state and measurements $\rho,A_a^{jk},A_j^\povm,B_j$, then the compressed operators satisfy
\begin{align}
\label{e:Bproj_main}
&\hat B_j^2 = \hat B_j \quad \forall j, \\
\label{e:Bcomplete_main}
&\sum_j \hat B_j = d I, \\
\label{e:Aproj_main}
&(\hat A^{jk}_1)^2 = \hat A^{jk}_1,\ 
(\hat A^{jk}_2)^2 = \hat A^{jk}_2,\ 
\hat A^{jk}_1\hat A^{jk}_2 = 0 \quad \forall j<k, \\
\label{e:BSIC_main}
&\hat B_j \hat B_k \hat B_j = s_{jk}\hat B_j \quad \forall j \neq k.
\end{align}
Moreover, up to local isometries, the state $\rho$ is a mixture of states of the form
\begin{align}\label{e:state_cert}
\bigoplus_\alpha \ket{\chi_\alpha}\otimes 
\ket{\varphi_{r_\alpha d}},\qquad \ket{\chi_\alpha}\in\C^{e_\alpha}\otimes \C^{f_\alpha}
\end{align}
where $e_\alpha, f_\alpha, r_\alpha \in \N$, and $\alpha$ is an element of a discrete index set. Up to the same local isometry on Alice's Hilbert space, the compressed elements of her $\povm$ setting, acting on $\bigoplus_\alpha \C^{e_\alpha} \otimes \C^{r_\alpha d}$, are given by
\begin{equation}\label{e:POVM_cert}
\hat A^\povm_j = \cN_j+\bigoplus_\alpha\left(
I_{e_\alpha}\otimes \frac{1}{d}\hC_{j,\alpha}
+W_{j,\alpha}\right)
\end{equation}
where $\tr_{\C^{r_\alpha d}}(W_{j,\alpha})=0$, and all diagonal $(\alpha,\alpha)$-blocks of $\cN_j$ are zero. The $\{ \hC_{j,\alpha} \}_j$ operators satisfy the same relations as the $\{ \hat B_j \}_j$ operators for every $\alpha$.
\end{prop}

The proof is based on a sum-of-squares decomposition of the shifted Bell operator (Appendix \ref{sec:bell}) and on the representation theory of an algebra related to the associated BIC-POVM (Appendix \ref{sec:repr}). The proof can be found in Appendix \ref{sec:entropy}.

Notice that we do not completely characterize the state up to local isometries. Instead, we certify the form in Eq.~\eqref{e:state_cert}. Here, the index $\alpha$ labels the irreducible representations of the C*-algebra generated by the relations \eqref{e:Bproj_main}, \eqref{e:Bcomplete_main} and \eqref{e:BSIC_main}. Both $\supp_A \rho$ and $\supp_B \rho$ decompose into a direct sum according to the irreducible representations, and we characterize the state on these subspaces. Notice in Eq.~\eqref{e:state_cert} that in every subspace we find a maximally entangled state of dimension $r_\alpha d$, tensored with some uncharacterized state $\ket{\chi_\alpha}$.

We also do not fully characterize the POVM $(A^\povm_j)_j$, as seen in Eq.~\eqref{e:POVM_cert}. The characterization is, however, sufficient for certifying randomness. This is because $\hat A^\povm_j$ acts trivially on the uncharacterized subspaces and subsystems of $\rho$, and acts like a BIC-POVM on the characterized part. Since on the characterized part the state is pure ($\ket{\varphi_{r_\alpha d}}$), no potential eavesdropper can be correlated to the state, since every purification of a pure state must be a product state. This intuition leads us to our main theorem:

\begin{thm}\label{t:main_main}
Suppose the state $\rho$ and measurement $(A^\povm_j)_j$ appear in an optimal quantum strategy for the Bell function \eqref{e:bell_operator}, and let $\ket{\psi}\in\cH_A\otimes\cH_B\otimes\cH_E$ be a purification of $\rho$. Then
\begin{equation}\label{e:mainthm}
\sum_{j=1}^{d^2} \ketbraq{j}_\uA 
\otimes \tr_{AB}\left[ \ketbraq{\psi} ( A^\povm_j \otimes I_B \otimes I_E)\right]=\left(\frac{1}{d^2}\sum_{j=1}^{d^2}\ketbraq{j}\right)\otimes
\sigma_E
\end{equation}
for some state $\sigma_E$ on $\cH_E$.

In particular, the maximal value of the Bell inequality \eqref{e:bell_operator} certifies $2 \log(d)$ bits of device-independent randomness from the outcome of the $\povm$ setting of Alice. Since the maximal violation can be achieved using a locally $d$-dimensional state, the theoretical maximum device-independent randomness, $2 \log(d)$ bits, can be achieved in every dimension~$d$.
\end{thm}
The proof of this theorem can be found in Appendix \ref{sec:entropy}. The value $2 \log(d)$ comes from computing the conditional von Neumann entropy of the state in Eq.~\eqref{e:mainthm}.
A locally $d$-dimensional realization leading to the maximal value $d^2$ is given by $\rho = \ketbraq{\varphi_d}$, $B_j = \ketbraq{e_j}$, where $\{ \frac1d \ketbraq{e_j} \}_{j=1}^{d^2}$ form a BIC-POVM such that $\abs{ \braket{e_j}{e_k} }^2 = s_{jk}$, $A^{jk}_{1(2)}$ being the transposition of the projection onto the eigenstate of $B_j-B_k$ with positive (negative) eigenvalue, and $A^{\povm}_j = \frac1d \ketbraq{e_j}^{\ti}$ (where all the transpositions are in the computational basis). In Proposition \ref{prop_Bell_certificate}, this realization corresponds to a single value for the index $\alpha$ (let us denote this value by $\tilde{\alpha}$), and $e_{\tilde{\alpha}} = f_{\tilde{\alpha}} = r_{\tilde{\alpha}} = 1$, $\cN_j = 0$, $ W_{j,\tilde{\alpha}} = 0$, and $C_{j,\tilde{\alpha}} = \ketbraq{e_j}^{\ti}$ for all~$j$.

\section{Discussion}

Our main result establishes that the previously known fundamental bound on device-independent randomness can be achieved in every dimension. To prove this, we designed explicit protocols that reach the fundamental bound. Since the maximal randomness requires (by definition) non-projective POVMs, standard self-testing arguments do not suffice for proving our result. As such, we developed new, weaker forms of self-testing, relying on the certification of compressed operators and the representation theory of measurement algebras. Our techniques are broadly applicable and we expect them to be highly useful in scenarios in which a full self-testing statement is not necessary, impossible, or impractical.

Let us address the practicality of physical implementations of non-projective measurements. Canonically, these are implemented by enlarging the Hilbert space with an ancillary system, and performing a joint projective measurement on the system and ancilla. As such, one might question whether our protocol is truly locally $d$-dimensional. Notice that importantly, the ancillary system is not shared between Alice and Bob but held only by Alice. Therefore, the entangled state is only locally $d$-dimensional. Moreover, in various standard physical implementations (e.g.~spatial modes of photons) the ancillary system is purely a mathematical model, and there is no need to introduce an actual new physical system (e.g.~a new photon) to implement a non-projective measurement \cite{GTZ+23,Goel2024}. Therefore, in the sense of `controllable degrees of freedom', our protocol can indeed be implemented using locally $d$ dimensions.

Further regarding the practicality of our setup, one might notice that while the protocol generates $2 \log(d)$ bits every time Alice measures $\povm$, she has in total $d^2(d^2-1)/2+1$ measurement settings, which need to be chosen randomly. Notice, however, that these settings do not need to be chosen \emph{uniformly} randomly. If Alice selects the setting $\povm$ most of the time, she can generate arbitrarily close to $2 \log(d)$ bits of randomness per measurement round. That is, she does not consume more randomness than what is generated. In fact, Alice only needs to consume a vanishing amount of randomness per measurement round. This approach is usually referred to as a spot-checking protocol \cite{spotchecking,spotchecking2,Liu2021}.

Our work opens up a couple of important future research directions, on top of the wider application of our certification methods. Following the discussion about the practicality, it would be desirable to work out the robustness of our protocol to experimental noise. The large freedom in the measurements (any BIC-POVM provides maximal randomness) signifies that implementations with a large noise tolerance are likely to exist. Since our certification is based on a Bell inequality (with quantifiable classical-quantum separation, see Appendix \ref{app:QCgap}), one possible avenue is to adapt robust self-testing techniques \cite{BLM+09,yang14,Kan16,Mancinska2024} to our certification methods. While it is plausible that this avenue is successful, quantifying the effect of noise on the representation-theoretic certification likely requires significant technical effort. As such, we leave the noise robustness quantification for a future study.

It would also be an interesting follow-up direction to figure out whether our protocols can be extended to certify maximal \emph{global} device-independent randomness. That is, $2 \log(d)$ bits on Alice's side and $2 \log(d)$ bits on Bob's side. Proving this would require extending our certification methods to a new setting on Bob's side, corresponding to a BIC-POVM.

Last, the algebraic characterization of BIC-POVMs in Appendix \ref{sec:icpovm} and \ref{a:icpovm} could provide new mathematical insights into the structure and famous existence problem of SIC-POVMs \cite{Zauner}, as SIC-POVMs are a class of BIC-POVMs. Further research into this algebraic characterization could also lead to other device-independent applications, similarly to what has been achieved with \emph{mutually unbiased measurements} \cite{tavakoli21,Colomer2022,farkas23}.

\newpage

\renewcommand{\appendixpagename}{\normalsize\centering  APPENDICES}
\begin{appendices}

\section{Fundamental bound on the device-independent randomness}
\label{app:bound}

The fundamental bound of $2 \log(d)$ bits of device-independent randomness from locally $d$-dimensional states can be argued by considering \textit{extremal} POVMs. In a fixed dimension, the set of POVMs forms a convex set. That is, one can take a convex combination of a POVM $(M_1, \ldots, M_m)$ and $(N_1, \ldots, N_m)$, given by $(\lambda M_1 + (1-\lambda)N_1, \ldots, \lambda M_m + (1-\lambda) N_m)$, which is again a POVM for every $\lambda \in [0,1]$. Furthermore, one can take convex combinations of POVMs with different numbers of outcomes, by appropriately padding the POVMs with zero operators. The set of POVMs endowed with this convex structure thus has extremal elements---POVMs that cannot be written as a non-trivial convex combination of two other POVMs. It is known from \cite{dariano_2005} that extremal POVMs in dimension $d$ can have at most $d^2$ non-zero elements. Consider then a correlation $p(a,b|x,y) = \tr[ (A^x_a \otimes B^y_b) \rho]$ that can be realized in dimension $d$. In particular, by some locally $d$-dimensional entangled state $\rho$, and $d$-dimensional POVMs $\{ (A^x_a)_a \}_x$ for Alice and $\{ (B^y_b)_b \}_y$ for Bob. All POVMs decompose into extremal ones. In particular, for the POVM $(A^x_a)_a$ which we use for randomness extraction, we have $A^x_a = \sum_{e=0}^{k-1} p_e A^{x,e}_a$, where $\{p_e\}_e$ is a probability distribution and $\{(A^{x,e}_a)_a\}_e$ are extremal POVMs and therefore have at most $d^2$ non-zero elements. Another realization of $p(a,b|x,y)$ is then given on the Hilbert space $ \bigoplus_{e=0}^{k-1} \C^d \otimes \C^d$. The realization consists of a locally $kd$-dimensional state $\rho' = \bigoplus_{e=0}^{k-1} p_e \rho$ and $kd$-dimensional POVMs $(\tilde{A}^x_a)_a = \left( \bigoplus_{e=0}^{k-1}A^{x,e}_a \right)_a$ and $(\tilde{B}^y_b)_b = \left( \bigoplus_{e=0}^{k-1} B^y_b \right)_b$. A tripartite extension of $\rho'$ is given by $\tilde{\rho} = \sum_{e=0}^{k-1} p_e \rho_e \otimes \ketbraq{e}$, where $\rho_e$ equals $\rho$ on the $e$-th copy of $\C^d \otimes \C^d$ and zero elsewhere, and $\{ \ket{e} \}_{e=0}^{k-1}$ is an orthonormal basis on Eve's Hilbert space. While we could in principle consider a purification of $\tilde{\rho}$ to be fully consistent with \eqref{eq:cqstate}, we will keep this mixed state for simplicity. Given this quantum realization of our observed correlation and the tripartite extension, the minimization of $H(\uA|E)_{\rho_{\uA E}}$ is upper bounded by $H(\uA|E)_{\tilde{\rho}_{\uA E}}$, where
\begin{align*}
\tilde{\rho}_{\uA E} &= \sum_a \ketbraq{a}_\uA \otimes \tr_{AB}[ \tilde{\rho} ( \tilde{A}^x_a \otimes I_B \otimes I_E )] \\
&= \sum_a \ketbraq{a}_\uA \otimes \left( \sum_{e=0}^{k-1} p_e \tr \left( \rho(A^{x,e}_a \otimes I_B) \right) \ketbraq{e} \right).
\end{align*}
The conditional entropy of this state is given by the conditional entropy of the classical distribution $q(a,e) = p_e \tr \left( \rho(A^{x,e}_a \otimes I_B) \right)$. That is,
\begin{align*}
H(\uA|E)_{\tilde{\rho}_{\uA E}} & = \sum_{e=0}^{k-1} p_e H(A|E = e) = \sum_{e=0}^{k-1} p_e H\left( \left\{ \tr \left( \rho(A^{x,e}_a \otimes I_B) \right) \right\}_a \right) \le \sum_{e=0}^{k-1} p_e \log(d^2) \\
& = 2 \log(d),
\end{align*}
where (with a slight abuse of notation) $H\left( \left\{ \tr \left( \rho(A^{x,e}_a \otimes I_B) \right) \right\}_a \right)$ is the Shannon entropy of the distribution $\left\{ \tr \left( \rho(A^{x,e}_a \otimes I_B) \right) \right\}_a $, which is upper bounded by $\log(d^2)$ due to the fact that at most $d^2$ elements of $(A^{x,e}_a)_a$ are non-zero. The theoretical maximum device-independent randomness certifiable using locally $d$-dimensional states is therefore $2\log(d)$.

\section{Balanced informationally complete POVMs}\label{sec:icpovm}

In this section we review the notion of (finite-dimensional) informationally complete POVMs and their construction following \cite{dariano_2004,dariano_2005}, and identify a sub-family of them that is the pillar of the scenario designed in this paper.
Throughout the paper, given a matrix $a$ we write $a^*$, $a^\ti$ and $\overline{a}$ to denote its conjugate transpose, transpose and complex conjugate, respectively.

If elements of a POVM, $M_1,\dots,M_m$ on $\C^d$, span $\mtxc{d}$, then $(M_j)_{j=1}^m$ is an \emph{informationally complete POVM (IC-POVM)} on $\C^d$ (note that $m\ge d^2$ in this case). By \cite[Corollary 6]{dariano_2005}, extremal $d^2$-outcome POVMs on $\C^d$ are IC-POVMs and consist of rank-one matrices.
In this paper we focus on a special family of $d^2$-outcome rank-one IC-POVMs on $\C^d$.
A \emph{balanced IC-POVM (BIC-POVM)} on $\C^d$ is $(\frac1d P_j)_{j=1}^{d^2}$ where $P_1,\dots,P_{d^2}$ are rank-one projections that form a basis for $\mtxc{d}$, and satisfy $\sum_{j=1}^{d^2}P_j=dI$. 
In the language of harmonic analysis, rank-one projections adding to a scalar multiple of the identity correspond to unit-norm tight frames \cite[Section 2]{waldron}.

\begin{lem}\label{l:matrix}
Let $(\frac1d P_j)_j$ be a BIC-POVM on $\C^d$.
Then the $d^2\times d^2$ matrix $S=(\tr(P_jP_k))_{j,k}$ is positive definite, $s_{jj}=1$ for all $j$, $0\le s_{jk}<1$ for all $j\neq k$, and $\sum_js_{jk}=d$ for all $k$.

Moreover, for every pair $(j,k)$ there exists a sequence $j=i_1,\dots,i_N=k$ in $[d^2]$ such that $s_{i_ni_{n+1}}\neq0$ for all $n=1,\dots,N-1$.
\end{lem}

\begin{proof}
The matrix $S$ is the Gram matrix of the basis $P_1,\dots,P_{d^2}$ with respect to the Frobenius inner product on $\mtxc{d}$, and therefore positive definite.
Since the $P_j$ are rank-one projections, the diagonal entries of $S$ equal 1. The off-diagonal entries of $S$ lie in $[0,1)$ by the Cauchy-Schwarz inequality and linear independence of the $P_j$. Furthermore, for every $k$ we have
\begin{equation*}
\sum_{j=1}^{d^2} s_{jk}
=\sum_{j=1}^{d^2} \tr(P_jP_k)
=\tr\left(\left(\sum_{j=1}^{d^2} P_j\right)P_k\right)=d.
\end{equation*}
To prove the second part of the statement, it suffices to show that the underlying graph of $S$ (i.e., the graph with vertices $[d^2]$, and an edge between $j$ and $k$ if and only if $s_{jk}\neq 0$) is connected. Suppose this is not the case. Then one can reorder the $P_j$ so that there exists $1\le m<d^2$ such that $P_jP_k=0$ for all $j\le m$ and $k>m$. Then there is a proper subspace $V$ of $\C^d$ such that $\ran P_j\subseteq V$ for all $j\le m$ and $\ran P_k\subseteq V^\perp$ for all $k>m$. But this contradicts the assumption that the $P_j$ span the whole $\mtxc{d}$.
\end{proof}

The matrix $S$ in Lemma \ref{l:matrix} is said to be \emph{induced} by the BIC-POVM. Its properties from Lemma \ref{l:matrix} are utilized frequently throughout the paper.

To show the existence of IC-POVMs it is common to use group-theoretic tools \cite{dariano_2004}. Consider the unitary projective representation of $\Z_d\times \Z_d$ on $\mathbb{P}(\C^d)$ given by the Weyl operators, {\it i.e.},~
\begin{align*}
U_{p,q}=\sum_{j=0}^{d-1}\omega^{jq} \ketbra{j\oplus p}{j}
\qquad\text{ for }(p,q)\in\Z_d\times\Z_d
\end{align*}
where  $\omega:= e^{\frac{2\pi i}{d}}$ is the principal $d$\textsuperscript{th} root of unity and $\oplus$ is addition modulo $d$. 
If $\ket{\psi}\in\C^d$ is a unit vector then
\begin{align}\label{eq:ICPOVM}
(M_{p,q})_{(p,q)\in\Z_d\times\Z_d},\quad\text{where}\quad    M_{p,q}:=\tfrac{1}{d} U_{p,q}\dyad{\psi} U_{p,q}^*
\end{align}
is a rank-one POVM. If furthermore $\ket{\psi}$
satisfies
\begin{align}\label{eq:req}
 \bra{\psi}U_{p,q}\ket{\psi}\neq 0   \quad \text{for all } p,q \in \Z_d,
\end{align}
then $(M_{p,q})_{(p,q)}$ is a rank-one IC-POVM by \cite[Section 3]{dariano_2004}, and furthermore a BIC-POVM.
Notice that if we consider an arbitrary state $\ket{\psi}=\sum_{j=0}^{d-1} \psi_j \ket{j}$, where $\{\ket{j}\}_{j=0}^{d-1}$ is the standard orthonormal basis on $\C^d$, and we define the subnormalized vector $\ket{\psi_p}:=\sum_{j=0}^{d-1}\overline{\psi_{j\oplus p}}\psi_j\ket{j}$, then
\begin{equation*}
\bra{\psi}U_{p,q}\ket{\psi} 
= \sum_{i,j,k=0}^{d-1} \omega^{jq} \psi_k\overline{\psi_i}   \braket{i}{j\oplus p}\braket{j}{k}
=\sum_{j=0}^{d-1} \omega^{jq} \overline{\psi_{j\oplus p}} \psi_{j}
=\bra{q}\mathcal{F}_d\ket{\psi_p}
\end{equation*}
where $\mathcal{F}_d$ denotes the quantum Fourier transform. The requirement in \eqref{eq:req} thus becomes
\begin{equation*}
    \bra{q} \cF_d \ket{\psi_p}\neq0 \quad \text{for all } p,q \in \Z_d.
\end{equation*}
To construct BIC-POVMs on $\C^d$ for every $d\in\N$, we thus need to identify some unit vectors $\ket{\psi}\in\C^d$ satisfying \eqref{eq:req}.
In \cite{dariano_2004,dariano_2005} it is asserted that states of the form $\ket{\psi}\propto \sum_k \alpha^{k}\ket{k}$,
where $\alpha\in\C$ and $|\alpha|<1$, fulfill the requirement \eqref{eq:req}. 
This assertion, however, needs some amendments:
if $d$ is even and $\alpha\in\R$, then $\ket{\psi}=\sum_k \alpha^{k}\ket{k}$ satisfies
$\bra{\psi}U_{d/2,1}\ket{\psi}=0$ by a direct calculation.
Thus we present the following sufficient criterion for $\ket{\psi}$ to satisfy \eqref{eq:req}.

\begin{prop}\label{p:goodvec}
Let $d\in\N$, $r\in (0,\frac12)$, and $t\in\R$ if $d$ is odd and $t\in\R\setminus \frac{1}{2d}\Z$ if $d$ is even.
Then $\ket{\psi}= \sum_{k=0}^{d-1} \alpha^{k}\ket{k}$ with $\alpha = re^{2\pi i t}$ satisfies
$\bra{\psi}U_{p,q}\ket{\psi}\neq0$ for all $p,q\in\Z_d$.
\end{prop}

\begin{proof}
We expand
\begin{align*}
\bra{\psi}U_{p,q}\ket{\psi}
&=\sum_{k,\ell=0}^{d-1}\sum_{j\in\Z_d}\omega^{jq}\overline{\alpha}^k\alpha^\ell
\braket{k}{j\oplus p}\!\braket{j}{\ell}\\
&=\sum_{\ell\le d-1-p}\omega^{\ell q}\overline{\alpha}^{\ell+p}\alpha^{\ell}
+\sum_{\ell\ge d-p}\omega^{\ell q}\overline{\alpha}^{\ell+p-d}\alpha^{\ell}\\
&=e^{-2\pi i pt}\left(
\sum_{\ell\le d-1-p}r^{2l+p}e^{2\pi i\frac{\ell q}{d}}
+\sum_{\ell \ge d-p}r^{2l+p-d}e^{2\pi i(\frac{\ell q}{d}+dt)}
\right).
\end{align*}
Therefore it suffices to see that
\begin{equation}\label{e:nonzero}
\sum_{\ell\le d-1-p}r^{2l+p}e^{2\pi i\frac{\ell q}{d}}
+\sum_{\ell \ge d-p}r^{2l+p-d}e^{2\pi i(\frac{\ell q}{d}+dt)}
\end{equation}
is nonzero for every $p,q\in\Z_d$. 
We shall reach this conclusion by observing that the largest absolute value of a term in \eqref{e:nonzero} strictly dominates the sum of the other absolute values.
Note that the largest absolute value equals $r^m$ where $m=\min\{p,d-p\}$, and the smaller absolute values belong to $\{r^{m+1},r^{m+2},\dots\}$ if $d$ is odd and to $\{r^{m+2},r^{m+4},\dots\}$ if $d$ is even.

First assume that $d$ is odd. 
Observe that the absolute values of terms in \eqref{e:nonzero} are pairwise distinct. Since
$r<\frac12$ implies
$$r^{m+1}+r^{m+2}+\cdots< \frac{r}{1-r}r^m < r^m,$$
we see that \eqref{e:nonzero} is nonzero.
Now assume that $d$ is even. Then the absolute values of terms in \eqref{e:nonzero} are not all pair-wise distinct; rather, some appear once and others appear twice.
If two terms in \eqref{e:nonzero} have the same absolute value, then one appears in the first sum and the other appears in the second sum, and the difference of their arguments is $2\pi(\frac{q}{2}+d t)$. 
Observe that $\frac{q}{2}+d t\notin \frac12\Z$ for every $q\in\{0,\dots,d-1\}$ by the assumption on~$t$. 
In particular, $z=e^{2\pi i(\frac{q}{2}+d t)}$ satisfies $|z|=1$, $z\neq -1$ and $z$ is a ratio of any distinct two terms in \eqref{e:nonzero} with the same absolute value.
Since $r<\frac12<\frac{1}{\sqrt{3}}$, we have
\begin{equation}\label{e:est}
2(r^{m+2}+r^{m+4}+\cdots) \le \frac{2r^2}{1-r^2}r^m < r^m.
\end{equation}
If $d\neq 2p$, then $r^m$ appears once as an absolute value in \eqref{e:nonzero}, and the sum of other absolute values in \eqref{e:nonzero} is strictly smaller than $r^m$ by \eqref{e:est}, so \eqref{e:nonzero} is nonzero.
If $d=2p$, then all the absolute values of terms in \eqref{e:nonzero} appear precisely twice, with the constant ratio $z\neq-1$. Thus, \eqref{e:nonzero} becomes $(1+z)(r^m e^{i\alpha_0} + r^{m+2} e^{i\alpha_2} + r^{m+4} e^{i\alpha_4} + \cdots)$ for some angles $\alpha_j$, and this expression is nonzero by \eqref{e:est}.

\end{proof}

In particular, by Lemma \ref{l:matrix} there exists a state in $\C^d$ satisfying \eqref{eq:req} for every $d\in\N$. Since \eqref{eq:req} is a topologically open condition, in fact a generic state in $\C^d$ satisfies \eqref{eq:req}.
Thus \eqref{eq:ICPOVM} can be used to produce a BIC-POVM from almost any state, which can then be employed in the scenario presented in Section \ref{sec:bell}.

Finally, for the sake of completeness, let us comment on the connection to a famous open problem in quantum information theory, {\it i.e.}, the existence problem of \emph{symmetric} IC-POVMs (SIC-POVMs). It is straightforward to show that if
\begin{align} \label{eq:cov_sic}
 \abs{\bra{q} \cF_d \ket{\psi_p}}^2=\frac{1}{d+1}   \quad\text{for all }(p,q)\neq(0,0),
\end{align}
then \eqref{eq:ICPOVM} is a SIC-POVM. Zauner famously conjectured in \cite{Zauner} that SIC-POVMs exist in every finite dimension $d$.
A state that fulfills \eqref{eq:cov_sic} is called a \textit{fiducial vector} and the SIC-POVM it induces is called covariant with respect to $\Z_d\times \Z_d$. Hence, one way to confirm Zauner's conjecture is to show the existence of fiducial vectors in every dimension. Fiducial vectors have been obtained in every dimension up to 67 \cite{Scott_2010}. 
Due to their uniform properties,
SIC-POVMs have been widely used to design protocols in quantum information theory, for example for quantum key distribution \cite{ren05},
quantum state tomography \cite{bent15},
entanglement detection \cite{shang18},
certification of non-projective measurements \cite{tav19}, and random number generation \cite{tavakoli21}.
On the other hand, the protocol of this paper relies on more general BIC-POVMs, whose existence is unproblematic (in contrast to SIC-POVMs).

\section{A Bell inequality}\label{sec:bell}

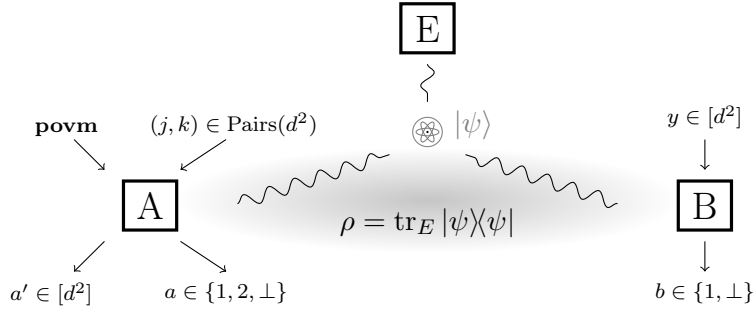
\begin{figure}
    \centering

\begin{tikzpicture}

\shadedraw[inner color=gray!45, outer color=white,draw=white]
(0,-.6) ellipse (3.3cm and .7cm);

\draw[black, very thick] (-3.5-.5,-1) rectangle (-2.8-.5,-.35) node[below,xshift=-.35cm,yshift=0cm] {{\large A}};
\draw[->] (-4.15-.5 , .2) -- (-3.75-.5 , -.2) node[above,yshift=.3cm,xshift=-.5cm] {\tiny{$\povm$}};
\draw[->] (-3.75-.5 , -1.15) -- (-4.15-.5 , -1.55) node[below,xshift=-.3cm] {\tiny{$a'\in [d^2]$}};
\draw[->] (-2.15-.5 , .2) -- (-2.75-.5 , -.2) node[above,yshift=.3cm,xshift=.7cm] {\tiny{$(j,k)\in \text{Pairs}(d^2)$}};
\draw[->] (-2.75-.5 , -1.15) -- (-2.15-.5 , -1.55) node[below,xshift=0cm] {\tiny{$a\in \{1,2,\perp \}$}};

\draw[black, very thick] (3.5+.5,-1) rectangle (2.8+.5,-.35) node[below,xshift=.35cm,yshift=0cm] {{\large B}};
\draw[->] (3.15+.5 , .2) -- (3.15+.5 , -.2) node[above,yshift=.4cm] {\tiny{$y\in [d^2]$}};
\draw[->] (3.15+.5 , -1.15) -- (3.15+.5 , -1.55) node[below,xshift=0cm] {\tiny{$b\in \{1,\perp \}$}};

\draw[black, very thick] (-.35,1.35) rectangle (.35,2) node[below,xshift=-.35cm,yshift=0cm] {{\large E}};

\def\science#1#2 {
\node[circle, draw, minimum size=2.7mm,color=gray] (#2) at #1 {};
\foreach \ang in {0,120,240}
  \draw[rotate around={\ang:#1},color=gray] #1 ellipse (1.5mm and .5mm);
\fill #1 circle (0.15mm);
}
\science{(0,.3)}{C1};

\path[draw,-,decorate,decoration={snake}] (-2.5,-.65) -- (-.5, 0);
\path[draw,-,decorate,decoration={snake}] (0,1.2) -- (0,.7)
node[right, xshift=.2cm,yshift=-.3cm] {\textcolor{gray}{\small{$\ket{\psi}$}}};
\path[draw,-,decorate,decoration={snake}] (2.5,-.65) -- (.5, 0) node at (0,-0.9) {\small{$\rho= \tr_E\dyad{\psi}$}};

\end{tikzpicture}

\label{fig:setup2_new}
\caption{The Bell scenario with malicious eavesdropper Eve who is trying to guess Alice's outcome of her $\povm$ setting.}
\end{figure}

In this paper we investigate a variation of the scenario presented in \cite{tavakoli21}. For any integer $d \ge 2$, we consider the following Bell scenario (see Figure \ref{fig:setup2_new}). Alice has $d^2(d^2-1)/2 + 1$ measurement settings, where the first $d^2(d^2-1)/2$ settings have 3 outcomes, and the last one has $d^2$ outcomes. The first $d^2(d^2-1)/2$ settings are labeled by pairs $(j,k)$ coming from the set
\begin{equation}
\text{Pairs}(d^2):=\{(j,k) \in [d^2] \times [d^2] \colon j<k\}
\end{equation}
The last of Alice's settings is labeled $\povm$. We label the outcomes of the 3-outcome measurements by $a \in \{1,2,\perp\}$ and the outcome of the $\povm$ measurement by $j \in [d^2]$. On the other side, Bob has $d^2$ measurement settings with two outcomes each, where the outcomes are labeled by $b \in \{1,\perp\}$.

The overall goal in this paper is to show that the outcome of Alice's $\povm$ setting is unpredictable for any eavesdropper Eve. We reach this conclusion in the device-independent setting, relying only on the observed correlation in the Bell scenario. We do this by constructing a Bell inequality whose maximal violation certifies certain desired properties of the measurements and the shared state. Since we aim to certify randomness coming from non-projective measurements, a fully assumption-free self-testing statement is impossible \cite[Theorem C]{baptista23}. As we will see, however, the maximal violation of our Bell inequality will certify enough about the state and the measurements to conclude that Eve cannot predict the outcome of Alice's $\povm$ measurement better than a random guess. 

We first simply state the Bell inequality, then the corresponding Bell operator and some intuition for why we chose this Bell inequality, in terms of the Bell operator. The inequality depends on $d$ and on a fixed BIC-POVM coming from unit vectors $\{\ket{e_j}\}_{j=1}^{d^2}$. The inequality depends on the parameters $s_{jk}=\abs{\braket{e_j}{e_k}}^2$ ({\it i.e.}, the entries of the matrix $S$ from Lemma \ref{l:matrix}), and is given by
\begin{equation}\label{e:bellfun}
\begin{split}
&  2\sum_{j<k}\sqrt{1-s_{jk}} \Big[ p(1,1|(j,k),j) + p(2,1|(j,k),k) - p(1,1|(j,k),k) - p(2,1|(j,k),j) \Big]  \\
& - \sum_{j<k}(1-s_{jk}) \big[ p_A(1|(j,k)) + p_A(2|(j,k)) \big]  - d(d-2) \sum_{j=1}^{d^2} p_B(1|j) - \sum_{j=1}^{d^2} p(j,\perp|\povm,j)
\end{split}
\end{equation}
where we use the shorthand notation $\sum_{j<k} = \sum_{j=1}^{d^2-1} \sum_{k = j+1}^{d^2}$. 
A quantum strategy for the above game consists of two finite-dimensional Hilbert spaces $\cH_A$ and $\cH_B$, a bipartite state $\rho$ on $\cH_A\otimes \cH_B$ given as a density matrix, POVMs $(A^{jk}_1,A^{jk}_2,I-A^{jk}_1-A^{jk}_2)$ for $j<k\in[d^2]$ and $(A^\povm_1,\dots,A^\povm_{d^2})$ on $\cH_A$ (Alice's measurements), and POVMs $(B_j,I-B_j)$ for $j\in[d^2]$ on $\cH_B$ (Bob's measurements).
The corresponding Bell operator is given by 
\begin{equation} \label{bell_ineq}
\begin{split}
W_d :=\, & 2\sum_{j<k} \sqrt{1-s_{jk}}(A^{jk}_1 - A^{jk}_2) \otimes (B_j - B_k)  -  \sum_{j<k} (1-s_{jk}) (A^{jk}_1 + A^{jk}_2)\otimes I \\  
&-  d(d-2)\sum_{j=1}^{d^2} I\otimes B_j - \sum_{j=1}^{d^2} A^\povm_j \otimes (I - B_j),    
\end{split}
\end{equation}
and the value attained by this quantum strategy is $\tr(W_d\rho)$.
Note that $W_d$ depends on the choice of a $d$-dimensional BIC-POVM, and not just on $d$; however, we index it with $d$ because its maximal quantum value depends on $d$ only, as seen in the remainder of the section.
In the following, we will often omit the tensor products with the identity operator in the notation if it does not cause confusion. That is, we write $A^{jk}_a$ instead of $A^{jk}_a \otimes I$ and $B_j$ instead of $I \otimes B_j$. 

The intuition for this Bell inequality comes from the aim of certifying a $d$-dimensional reference strategy (see section \ref{subsec:ref_strategy}). In this strategy, we use the maximally entangled state $\ket{\varphi_d}$, and the $B_j$ operators are rank-1 projections, and therefore $B_j - B_k$ is rank-2. In order to certify this operator, we set $A^{jk}_{1(2)}$ to be the transposition of the projection onto the positive (negative) eigenvalue of $(B_j - B_k)$. Since the dimension is in general greater than 2, Alice's $jk$ measurement must have a third outcome, occupying the remaining $d-2$ dimensions of her Hilbert space. That this third outcome occurs is enforced by the second term in $W_d$, penalizing Alice for outputting 1 or 2. The last term in $W_d$ ensures a high level of correlation between Alice's $\povm$ setting and Bob's measurements. The overlaps between Bob's measurements are certified using the $s_{jk}$ dependence in the coefficients in $W_d$. The precise form of these coefficients is a result of a simple sum-of-squares decomposition (see section \ref{subsec:SOS}).

\subsection{Reference strategy}\label{subsec:ref_strategy}
We now consider a reference strategy in the quantum model reaching a score of $d^2$. Afterwards we proceed to show that $d^2$ is indeed optimal. Let Alice and Bob share the canonical maximally entangled state
\begin{equation*}
    \ket{\varphi_d}:=\frac{1}{\sqrt{d}}\sum_{k=0}^{d-1} \ket{kk}.
\end{equation*}
Bob's measurements are given by $\Tilde{B}_j:=\dyad{\Tilde{\psi}_j}$ and $\Tilde{B}_j^{\perp}:=I-\dyad{\Tilde{\psi}_j}$, where $\{\Tilde{\psi}_j\}_j$ are unit vectors inducing the BIC-POVM. Now, the operator $\Tilde{B}_j-\Tilde{B}_k$ for $j\neq k$ is traceless, rank 2 and Hermitian. Hence, it has a spectral decomposition of the form 
\begin{equation*}
    \Tilde{B}_j-\Tilde{B}_k=\gamma (\dyad{a^{jk}_1}-\dyad{a^{jk}_2})
\end{equation*}
where $\ket{a_1^{jk}}, \ket{a_2^{jk}}$ are orthogonal unit vectors and $\gamma \in (0,2)$. Furthermore,
$$2\gamma^2
=\tr\left((\Tilde{B}_j-\Tilde{B}_k)^2\right)
=\tr\left(\Tilde{B}_j+\Tilde{B}_k-\Tilde{B}_j\Tilde{B}_k-\Tilde{B}_k\Tilde{B}_j\right)
=2-2s_{jk}$$
so $\gamma=\sqrt{1-s_{jk}}$.
Let Alice's POVM associated with her 3-output measurement be given by $\Tilde{A}_1^{jk}:=(\dyad{a_1^{jk}})^{\ti}$, $\Tilde{A}_2^{jk}:=(\dyad{a_2^{jk}})^{\ti}$ and $\Tilde{A}_{\perp}^{jk}:=(I-\dyad{a_1^{jk}}-\dyad{a_2^{jk}})^{\ti}$. Notice in particular that we have the relations
\begin{equation} \label{int_sec31}
\begin{split}
    \Tilde{B}_{j}-\Tilde{B}_k&=\sqrt{1-s_{jk}}(\Tilde{A}^{jk}_1- \Tilde{A}^{jk}_2)^{\ti}
    \\
    (\Tilde{B}_{j}-\Tilde{B}_k)^2&=(1-s_{jk})(\Tilde{A}^{jk}_1+ \Tilde{A}^{jk}_2)^{\ti}
\end{split}
\end{equation}
Moreover, let $\Tilde{A}^{\povm}_j:=\tfrac{1}{d} \Tilde{B}_j^{\ti}$.

Recall that since $\ket{\varphi_d}$ is the canonical maximally entangled state we have for any $X \in M_d(\C)$ that $X\otimes I \ket{\varphi_d}=I \otimes X^\ti \ket{\varphi_d}$. Then $\bra{\varphi_d}W_d\ket{\varphi_d}$, 
the value of the Bell function \eqref{e:bellfun} at this strategy,
\begin{equation*}
\begin{split}
&2\sum_{j<k} \sqrt{1-s_{jk}} \bra{\varphi_d}(\Tilde{A}^{jk}_1 - \Tilde{A}^{jk}_2) \otimes (\Tilde{B}_j - \Tilde{B}_k)\ket{\varphi_d}  -  \sum_{j<k} (1-s_{jk}) \bra{\varphi_d}(\Tilde{A}^{jk}_1 + \Tilde{A}^{jk}_2)\ket{\varphi_d} \\  
&-  d(d-2)\sum_{j=1}^{d^2} \bra{\varphi_d}I\otimes\Tilde{B}_j\ket{\varphi_d} - \sum_{j=1}^{d^2} \bra{\varphi_d}\Tilde{A}^\povm_j \otimes (I - \Tilde{B}_j) \ket{\varphi_d},
\end{split}
\end{equation*}
reduces to
\begin{equation*}
\begin{split}
&\sum_{j<k}  \bra{\varphi_d}I \otimes (\Tilde{B}_j - \Tilde{B}_k)^2\ket{\varphi_d}  -  d^2(d-2) \bra{\varphi_d}I \otimes I\ket{\varphi_d}
\\
=\, & \tfrac{1}{d}\sum_{j<k}  \tr[(\Tilde{B}_j - \Tilde{B}_k)^2]  -  d^2(d-2)
=\frac{2}{d}\sum_{j<k}  (1-s_{jk})  -  d^2(d-2)
\end{split}
\end{equation*}
using the relations \eqref{int_sec31}, $\sum_j \Tilde{B}_j=dI$ and $(\Tilde{A}^{\povm}_j)^{\ti}(I-\Tilde{B}_j)=0$.
The sum $\sum_{j<k}s_{jk}$ is given by the sum of the upper triangle of the matrix $S$ induced by the 
initial BIC-POVM. By Lemma \ref{l:matrix} we have $\sum_{j<k}s_{jk}=\tfrac{d^3-d^2}{2}$ 
so $\bra{\varphi_d}W_d\ket{\varphi_d}$ further simplifies to
\begin{equation*}
    \frac{2}{d} \left(\frac{d^4-d^2}{2} - \frac{d^3-d^2}{2}\right)   -  d^2(d-2)=d^2,
\end{equation*}
as desired.  

\subsection{Sum-of-squares decomposition of the Bell inequality}\label{subsec:SOS}

In order to show that $d^2$ is the maximal quantum value of \eqref{e:bellfun}, or equivalently, a tight upper bound on the eigenvalues of the Bell operator $W_d$, we will need the following technical lemma.

\begin{lem}\label{l:little}
Let $X,Y\succeq0$ satisfy $X+Y\preceq I$. Then $X+Y-(X-Y)^2\succeq0$. Furthermore, equality holds if and only if $X,Y$ are projections orthogonal to each other.
\end{lem}

\begin{proof}
The inequality follows from rewriting $X+Y-(X-Y)^2$ as
$$(I-X+Y)X(I-X+Y)+(I+X-Y)Y(I+X-Y)+(X-Y)(I-X-Y)(X-Y),$$
where all three terms are positive semidefinite by construction.

Now suppose $X,Y\succeq0$, $X+Y\preceq I$ and $X+Y-(X-Y)^2=0$. The above certificate for positive semidefinitenss of $X+Y-(X-Y)^2$ implies
$$(I-X+Y)X=0,\qquad (I+X-Y)Y=0.$$
In particular, $XY$ is hermitian, so $X$ and $Y$ commute. Therefore they are jointly diagonalizable, with corresponding diagonal entries $x_i$ and $y_i$, which satisfy for all $i$
\begin{equation}\label{e:eig_sys}
x_i,y_i\ge0,\quad x_i+y_i\le 1,\quad (1-x_i+y_i)x_i=0,\quad (1+x_i-y_i)y_i=0.
\end{equation}
A direct calculation shows that the solutions of \eqref{e:eig_sys} for a fixed $i$ are $(0,0),(1,0),(0,1)$. Therefore, $X$ and $Y$ are projections and $XY=0$.
\end{proof}

\begin{prop}\label{p:bellmax}
For any integer $d \ge 2$ and any choice of POVMs $\{A^{jk}_a\}_{jk}$, $\{A^{\povm}_j\}$ and $\{B_j\}$ in the above scenario, we have
\begin{equation*}
    d^2 I - W_d \succeq 0.
\end{equation*}
\end{prop}
\begin{proof}
    Let 
\begin{equation}\label{e:SOS_new_finish}
\begin{split}
\Theta_d :=\, &  \sum_{j<k}  \left[\sqrt{1-s_{jk}} (A^{jk}_1 - A^{jk}_2) - ( B_j- B_k ) \right]^2 + \left( d I - \sum_{j=1}^{d^2} B_j \right)^2 \\
& \left. + \sum_{j=1}^{d^2} A^\povm_j \otimes (I - B_j) + d^2 \sum_{j=1}^{d^2} \left( B_j - B_j^2 \right) \right. \\
& \left. +  \sum_{j<k}(1-s_{jk}) \left[ A^{jk}_1 + A^{jk}_2 - \left( A^{jk}_1 - A^{jk}_2 \right)^2 \right].
\right.
\end{split}
\end{equation}
The goal is to show that $W_d+\Theta_d=d^2I$. First notice that if we expand the square in the first term of $\Theta_d$ we get
\begin{equation} \label{int_s1}
\begin{split}
    &\left[\sqrt{1-s_{jk}} (A^{jk}_1 - A^{jk}_2) - ( B_j- B_k ) \right]^2\\ 
    = \, & (1-s_{jk}) (A^{jk}_1 - A^{jk}_2)^2 - 2\sqrt{1-s_{jk}} (A^{jk}_1 - A^{jk}_2)\otimes (B_j - B_k) 
    \\ &+(B_j^2+B_k^2-\{ B_j,B_k\}).
\end{split}
\end{equation}
The second term on the right hand side in \eqref{int_s1} cancels with the first term of $W_d$ in \eqref{bell_ineq}. The first term on the right hand side of \eqref{int_s1} cancels with the last part of the last term of~$\Theta_d$. Notice also that the second term of $W_d$ in \eqref{bell_ineq} cancels with the first part of the last term of~$\Theta_d$, and that the terms involving $A^{\povm}_j$ also cancel. Altogether we have reduced $W_d+\Theta_d$ to 
\begin{equation*}
\begin{split}
    \sum_{j<k} (B_j^2+B_k^2-\{ B_j,B_k\}) +  \left( d I - \sum_{j=1}^{d^2} B_j \right)^2  + d^2 \sum_{j=1}^{d^2} \left( B_j - B_j^2 \right)-  d(d-2)\sum_{j=1}^{d^2} B_j.
\end{split}
\end{equation*}
Using the fact that
\begin{align*}
    \sum_{j<k} (B_j^2+B_k^2)=\tfrac{1}{2}\left(\sum_{j,k=1}^{d^2}(B_j^2+B_k^2)-2\sum_{\ell=1}^{d^2}B_{\ell}^2 \right)=(d^2-1)\sum_{j=1}^{d^2} B_j^2,
\end{align*}
we get 
\begin{equation}\label{int_s2}
    \begin{split}
      &(d^2-1)\sum_{j=1}^{d^2} B_j^2-\sum_{j<k}\{ B_j,B_k\}+\left( d I - \sum_{j=1}^{d^2} B_j \right)^2 \\ 
      &+ d^2 \sum_{j=1}^{d^2} \left( B_j - B_j^2 \right) -  d(d-2)\sum_{j=1}^{d^2} B_j   .
    \end{split}
\end{equation}
Expanding the third term leads to
\begin{equation}\label{int_s3}
    \begin{split}
     \left( d I - \sum_{j=1}^{d^2} B_j \right)^2&= d^2I+ \sum_{j,k=1}^{d^2}B_jB_k -2d\sum_{j=1}^{d^2} B_j \\
    &=d^2I+ \sum_{j=1}^{d^2} B_j^2 + \sum_{j<k} \{B_j,B_k \}-2d\sum_{j=1}^{d^2} B_j ,
    \end{split}
\end{equation}
where we have used
\begin{align*}
\sum_{j,k=1}^{d^2}B_jB_k=\sum_{j}B_j^2+\sum_{j<k} \{B_j,B_k \}.
\end{align*}
Upon inserting \eqref{int_s3} in \eqref{int_s2} one obtains $d^2I$ after a straightforward simplification. We have thus shown $d^2 I - W_d=\Theta_d$. 

Notice that the first two terms of $\Theta_d$ are squares of Hermitian operators which means that they are positive semidefinite. The third and the fourth terms are also positive semidefinite which follows from the fact that the operators form POVMs. It follows by Lemma \ref{l:little} that the last term of $\Theta_d$ is positive semidefinite as well. We conclude that $\Theta_d$ is positive semidefinite and therefore $d^2 I - W_d \succeq 0$, as desired.
\end{proof}

\section{Representation-theoretic auxiliaries}\label{sec:repr}

Before analyzing strategies where the Bell function \eqref{e:bellfun} attains the maximal quantum value $d^2$, we require a few intermediate results on operators satisfying the fundamental relations of BIC-POVMs, and a decomposition of states into maximally entangled ones. These results are obtained using techniques from C*-algebras and representation theory (for the general theory, see \cite{Tak02} and \cite{procesi}), and might be of independent interest.

\subsection{A relevant C*-algebra}\label{sec:cstar}

This subsection introduces a C*-algebra whose representations lurk behind optimal strategies for the scenario in Section \ref{sec:bell}. 
Let $S\in\mtxr{d^2}$ be a matrix induced by a BIC-POVM. With it we associate the universal unital C*-algebra
$$\cA_S={\rm C}^*\bigg\langle x_1,\dots,x_{d^2}\colon 
x_j=x_j^*=x_j^2\ \forall j,\ \sum_{j=1}^{d^2}x_j=d,\ 
x_jx_kx_j=s_{jk}x_j\ \forall j, k\bigg\rangle.
$$
For a finite-dimensional representation $\pi:\cA_S\to\mtxc{D}$ denote $\dim\pi=D$.

\begin{prop}\label{p:rk}
Let $\pi$ be a finite-dimensional representation of $\cA_S$. Then $\dim \pi$ is divisible by $d$, $\tr \pi(x_j)=\frac{\dim\pi}{d}$ for all $j$, and $\pi(x_1),\dots,\pi(x_{d^2})$ are linearly independent.
\end{prop}

\begin{proof}
Denote $X_j=\pi(x_j)$. For all $j\neq k$ we have $s_{jk}\tr(X_j)=\tr(X_jX_k)=s_{jk}\tr(X_k)$, and thus $\tr(X_j)=\tr(X_k)$ whenever $s_{jk}\neq0$. By Lemma \ref{l:matrix} it follows that $\tr(X_1)=\cdots=\tr(X_{d^2})$. Then for all $j$,
$$d^2\tr(X_j)=\sum_{k=1}^{d^2}\tr(X_k)=d\tr I$$
and so $d\tr (X_j)=\tr I=\dim\pi$.
The Gram matrix of $\pi(x_1),\dots,\pi(x_{d^2})$ with respect to the Frobenius inner product on $\mtxc{\dim\pi}$ equals 
$(\tr(X_jX_k))_{j,k}=(\frac{\dim\pi}{d}s_{jk})_{j,k}=\frac{\dim\pi}{d}S$, which is invertible by Lemma \ref{l:matrix}. Therefore $\pi(x_1),\dots,\pi(x_{d^2})$ are linearly independent.
\end{proof}

Representations of $\cA_S$ of minimal dimension mimic the properties of BIC-POVMs in the following sense.

\begin{prop}\label{p:lowdim}
BIC-POVMs that induce $S$ correspond to the $d$-dimensional representations of $\cA_S$.
\end{prop}

\begin{proof}
If a BIC-POVM $\frac{1}{d}X_1,\dots,\frac{1}{d}X_{d^2}$ induces $S$, then the tuple $(X_j)_j$ gives rise to a $d$-dimensional representation $\pi$ of $\cA_S$.
On the other hand, if $\pi$ is a $d$-dimensional representation of $\cA_S$, then $\rk\pi(x_j)=\tr\pi(x_j)=1$ by Proposition \ref{p:rk}. Consequently $\tr(\pi(x_j)\pi(x_k))=s_{jk}$ for all $j,k$, and so $\frac1d\pi(x_1),\dots,\frac1d\pi(x_{d^2})$ is a BIC-POVM inducing $S$.
\end{proof}

By Proposition \ref{p:lowdim}, every BIC-POVM inducing $S$ gives rise to an irreducible representation of $\cA_S$. On the other hand, $\cA_S$ might have other irreducible representations (cf. Section \ref{a:icpovm}).
There is also an alternative definition of $\cA_S$.

\begin{lem}\label{l:alt}
The C*-algebra $\cA_S$ equals
$${\rm C}^*\bigg\langle x_1,\dots,x_{d^2}\colon 
x_j=x_j^*=x_j^2\ \forall j,\ \sum_{j=1}^{d^2}x_j=d,\ 
(1-s_{jk})(x_j-x_k)=(x_j-x_k)^3
\ \forall j,k\bigg\rangle.$$
\end{lem}

\begin{proof}
Let $\widetilde{\cA}_S$ denote the new C*-algebra from the statement of Lemma \ref{l:alt}.
Since
$$(x_j-x_k)^3=(x_j-x_k)-(x_jx_kx_j-x_kx_jx_k)$$
for projections $x_j,x_k$, it follows that the relations $(1-s_{jk})(x_j-x_k)=(x_j-x_k)^3$ in $\widetilde{\cA}_S$ can be replaced with
\begin{equation}\label{e:expanded}
x_jx_kx_j-s_{jk}x_j=x_kx_jx_k-s_{jk}x_k.
\end{equation}
Thus $\cA_S$ is clearly a quotient of $\widetilde{\cA}_S$. It now suffices to see that $x_jx_kx_j-s_{jk}x_j=0$ in $\widetilde{\cA}_S$.
Firstly, observe that $x_jx_kx_j-s_{jk}x_j$ is positive semidefinite 
in $\widetilde{\cA}_S$ for $j\neq k$ since
$$x_jx_kx_j-s_{jk}x_j=\left(
\tfrac{1}{\sqrt{1-s_{jk}}}x_jx_kx_j-\tfrac{s_{jk}}{\sqrt{1-s_{jk}}}x_j
\right)^2$$
holds by \eqref{e:expanded}.
Next, for every $j$ we have
$$
\sum_{k}\left(x_jx_kx_j-s_{jk}x_j\right)
=x_j\left(\sum_kx_k\right)x_j-\left(\sum_{k}s_{jk}\right)x_j
=x_j\cdot d\cdot x_j-dx_j=0.
$$
Therefore $x_jx_kx_j-s_{jk}x_j=0$ for all $j,k$ by semidefiniteness, as desired.
\end{proof}

\subsection{Local support of a mixed bipartite state}

In this subsection we recall the definition of the local support of a mixed bipartite state, and highlight some of its features.

Given a mixed bipartite state $\rho\in\mtxc{d_A}\otimes\mtxc{d_B}$, its \emph{local support} on Alice's side $\supp_A\rho\subseteq\C^{d_A}$ is the range of $\tr_B\rho\in\mtxc{d_A}$; here, $\tr_B$ denotes the partial trace over $\C^{d_B}$.
The local support of $\rho$ on Bob's side 
$\supp_B\rho\subseteq\C^{d_B}$ is defined analogously.
Given an operator $X\colon\C^{d_A}\to\C^{d_A}$, its \emph{compression} onto the local support of $\rho$ is the operator $\hat X\colon \supp_A\rho\to \supp_A\rho$ given by $\hat X=U^*XU$ where $U\colon \supp_A\rho\to\C^{d_A}$ is the inclusion
(that is, $U^*U$ is the identity on $\supp_A\rho$, and $UU^*$ is the projection acting on $\C^{d_A}$ whose range is $\supp_A\rho$).
Analogously we define compressions onto $\supp_B\rho$ for operators on $\C^{d_B}$.
The following lemma might be folklore (especially (i)), but we record it for the sake of completeness.

\begin{lem}\label{l:supp}
Let $\rho\in\mtxc{d_A}\otimes\mtxc{d_B}$ be a mixed bipartite state.
\begin{enumerate}[(i)]
\item If $U\colon \supp_A\rho\to\C^{d_A}$ is the inclusion then $(UU^*\otimes I)\rho=\rho$.
\item If $Y\in\mtxc{d_B}$ then $\ran \tr_B\big((I\otimes Y)\rho\big)\subseteq \supp_A \rho$.
\end{enumerate}
\end{lem}

\begin{proof}
First we check (i) and (ii) for a pure state $\rho$. After a local unitary basis change we can assume that $\rho=\sum_{i,j=0}^{r-1} \lambda_i\lambda_j \ketbra{ii}{jj}$ for $\lambda_i,\lambda_j>0$. 
Then $\tr_B\rho=\sum_{i=0}^{r-1} \lambda_i^2 \dyad{i}$, so $UU^*=\sum_{i=0}^{r-1}\dyad{i}$ and (i) holds.
The range of $\tr_B((I\otimes Y)\rho)=\sum_{i,j=0}^{r-1} \lambda_i\lambda_j \bra{j}Y\ket{i}\ketbra{i}{j}$ is contained in the span of $\{\ket{0},\dots,\ket{r-1}\}$, so (ii) holds.

If $\rho=\sum_k\gamma_k\rho_k$ where $\gamma_k>0$ and $\rho_k$ are pure states, then $\supp_A\rho=\sum_k\supp_A\rho_k$ 
by semidefiniteness. Thus $(UU^*\otimes I)\rho_k=\rho_k$ for all $k$. Also, the range of $\tr_B((I\otimes Y)\rho)=\sum_k \gamma_k\tr_B((I\otimes Y)\rho_k)$ is contained in $\supp_A\rho$ by the previous paragraph. Therefore (i) and (ii) hold for $\rho$.
\end{proof}

\subsection{Block-wise maximally entangled states}\label{sec:max_ent}

In this subsection we see how a synchronicity condition on a mixed bipartite state $\rho$ implies that $\rho$ admits a block diagonal decomposition into maximally entangled states. In particular, Proposition \ref{p:maxent} below is later applied to the state of an optimal strategy for the scenario in Section \ref{sec:bell}.
Within the following proposition 
and the rest of the paper,
we often tacitly reshuffle the order of tensor factors for the sake of notation,

\begin{prop}\label{p:maxent}
Let $E_1,\dots,E_n\in\mtxc{d_A}$ and $F_1,\dots,F_n\in\mtxc{d_B}$ be hermitian matrices, and $\rho\in\mtxc{d_A}\otimes\mtxc{d_B}$ a mixed bipartite state, such that
\begin{equation}\label{e:sync}
(E_j\otimes I)\rho=(I\otimes F_j)\rho \qquad \text{for }j\in[n].
\end{equation}
Then there exist $L\in\N$, $d_\alpha,e_\alpha,f_\alpha\in\N$ for $\alpha \in [L]$, and isometries $U:\bigoplus_\alpha\C^{e_\alpha}\otimes\C^{d_\alpha}\to\C^{d_A}$, $V:\bigoplus_\alpha\C^{f_\alpha}\otimes\C^{d_\alpha}\to\C^{d_B}$ such that:
\begin{enumerate}[(i)]
\item $\ran U=\supp_A\rho$ and $\ran V=\supp_B\rho$;
\item $U^*E_jU\in \bigoplus_\alpha I_{e_\alpha}\otimes\mtxc{d_\alpha}$ and $V^*F_jV\in \bigoplus_\alpha I_{f_\alpha}\otimes\mtxc{d_\alpha}$ for $j\in [n]$;
\item $(U\otimes V)^*\rho(U\otimes V)$ is a mixture of pure states of the form
\begin{align*}
\bigoplus_\alpha \ket{\chi_\alpha}\otimes 
\ket{\varphi_{d_\alpha}}
&\in
\bigoplus_\alpha \left(\C^{e_\alpha}\otimes \C^{f_\alpha}\right)\otimes 
\left(\C^{d_\alpha}\otimes \C^{d_\alpha}\right) \\
&\subset
\left(\bigoplus_\alpha\C^{e_\alpha}\otimes\C^{r_\alpha d}\right)\otimes
\left(\bigoplus_\alpha\C^{f_\alpha}\otimes\C^{r_\alpha d}\right);
\end{align*}
\item if $\rho$ is pure then $U^*E_jU=\left(V^*F_jV\right)^\ti$ for $j\in[n]$.
\end{enumerate}
\end{prop}

\begin{proof}
Let $\hat{E}_j$ and $\hat{F}_j$ denote the compressions of $E_j$ onto $\supp_A\rho$ and $F_j$ onto $\supp_B\rho$, respectively. Also, let $\hat\rho$ denote the compression of $\rho$ onto $\supp_A\rho\otimes\supp_B\rho$.
Since $\hat{E}_j,\hat{F}_j$ are hermitian, the unital algebras generated by $\hat{E}_1,\dots,\hat{E}_n$ and $\hat{F}_1,\dots,\hat{F}_n$ are finite-dimensional C*-algebras. Let us identify $\supp_A\rho=\C^{N_A},\supp_B\rho=\C^{N_B}$. By \cite[Theorem I.11.2]{Tak02} there exist unitaries $U\in\mtxc{N_A},V\in\mtxc{N_B}$ such that
\begin{equation}\label{e:semisimple}
\begin{split}
\check{E_j}:=U^*\hat{E}_jU
&=\bigoplus_{\alpha=1}^L X_{j,\alpha}^{\oplus e_\alpha}
\oplus \bigoplus_{\alpha=1}^{M'} {X'}_{j,\alpha}^{\oplus g_\alpha'}
\qquad \text{for }j\in [n],\\
\check{F_j}^\ti:=(V^*\hat{F}_jV)^\ti
&=\bigoplus_{\alpha=1}^L X_{j,\alpha}^{\oplus f_\alpha}
\oplus \bigoplus_{\alpha=1}^{M''} {X''}_{j,\alpha}^{\oplus g_\alpha''}
\qquad \text{for }j\in [n],
\end{split}
\end{equation}
where $(X_{j,\alpha})_j,({X'}_{j,\alpha})_j,({X''}_{j,\alpha})_j$ are pairwise unitarily non-equivalent irreducible tuples, and $X_{j,\alpha}\in \mtxc{d_\alpha}$.

Let $\check{\rho}:=(U\otimes V)^*\hat{\rho}(U\otimes V)=\sum_{k=1}^K p_k\dyad{\psi_k}$ be a spectral decomposition of $\check{\rho}$, where $\ket{\psi_1},\dots,\ket{\psi_K}\in\C^{N_A}\otimes\C^{N_B}$ are orthogonal states, and $p_1,\dots,p_K>0$. 
Let $\mat$ denote the matricization operator transforming tensors into matrices, determined by $\mat(\ket{ab})=\ket{a}\!\bra{b}$. 
Then
\begin{equation}\label{e:kernels}
\bigcap_{k=1}^K\ker \mat(\ket{\psi_k})=0,\qquad
\bigcap_{k=1}^K\ker \mat(\ket{\psi_k})^\ti=0
\end{equation}
since all operators and states have been compressed to $\supp_A\rho$ and $\supp_B\rho$.
By \eqref{e:sync} we have $(\check{E}_j\otimes I)\ket{\psi_k}=(I\otimes \check{F}_j)\ket{\psi_k}$, and therefore
\begin{equation}\label{e:commut_psi}
\check{E}_j\mat(\ket{\psi_k})=\mat(\ket{\psi_k})\check{F}_j^\ti \qquad	\text{for }j\in [n],\ k\in [K].
\end{equation}
Suppose $M'>0$ holds in \eqref{e:semisimple}. 
By \eqref{e:commut_psi} and Schur's lemma \cite[Corollary 6.1.7]{procesi}, the $(L+1)$\textsuperscript{th} block-row of $\mat(\ket{\psi_k})$ is zero for every $k$, which contradicts \eqref{e:kernels}. Therefore $M'=0$, and analogously $M''=0$.
Hence
\begin{equation}\label{e:semisimple2}
\check{E_j}=\bigoplus_{\alpha=1}^L X_{j,\alpha}^{\oplus e_\alpha},
\qquad
\check{F_j}^\ti=\bigoplus_{\alpha=1}^L X_{j,\alpha}^{\oplus f_\alpha},
\end{equation}
and we can view $\mat(\ket{\psi_k})$ as a block matrix with $\sum_\alpha e_\alpha$ block-rows and $\sum_\alpha f_\alpha$ block-columns according to decompositions \eqref{e:semisimple2}. By \eqref{e:commut_psi}, irreducibility of $(X_{j,\alpha})_j$ and another application of Schur's lemma, a block in $\mat(\ket{\psi_k})$ is a nonzero scalar multiple of the identity matrix if the row and column correspond to the same $(X_{j,\alpha})_j$, and zero otherwise. Hence $\mat(\ket{\psi_k})=\bigoplus_\alpha R_{k,\alpha}\otimes I_{d_\alpha}$ for some $R_{k,\alpha}\in\C^{e_\alpha\times f_\alpha}$, and therefore
$$\check\rho=\sum_{k=1}^K p_k 
\dyad{\psi_k},\qquad
\ket{\psi_k}=\bigoplus_{\alpha=1}^L \sqrt{d_\alpha}\mat^{-1}(R_{k,\alpha})\otimes \ket{\varphi_{d_\alpha}}.
$$
Finally, assume that $\rho$ is pure. Then $\check\rho=\dyad{\psi_1}$ and $\mat(\psi_1)=\bigoplus_\alpha R_{1,\alpha}\otimes I_{d_\alpha}$ is invertible by \eqref{e:kernels}.
This is only possible if $R_{1,\alpha}$ is invertible for every $\alpha$. In particular, $R_{1,\alpha}$ has to be a square matrix, and therefore $e_\alpha=f_\alpha$ for all $\alpha$. Hence $\check{E}_j=\check{F}_j^\ti$ by \eqref{e:semisimple2}.
\end{proof}

\begin{rem}
Proposition \ref{p:maxent}(iv) is a special case of \cite[Corollary 3.6]{Mancinska2024}. On the other hand, purity in (iv) is essential. A counterexample with $n=1$ and $d_A=d_B=3$ is given by
$E_1=1\oplus1\oplus0$, $F_1=1\oplus0\oplus0$ and $\rho=\frac12(\dyad{\psi_1}+\dyad{\psi_2})$
where $\psi_1=\frac1{\sqrt{2}}(\ket{00}+\ket{21})$ and
$\psi_2=\frac1{\sqrt{2}}(\ket{10}+\ket{22})$.
\end{rem}

\section{Optimal strategy analysis}\label{sec:entropy}

This section analyzes optimal strategies for the scenario in Section \ref{sec:bell}, and establishes our main result on certifying maximal randomness (Theorem \ref{t:main}).
Throughout the section let $\cH_A$ and $\cH_B$ be finite-dimensional Hilbert spaces, $\rho$ a mixed bipartite state on $\cH_A\otimes \cH_B$, $A_1^{jk},A_2^{jk},A_j^\povm$ for $(j,k)\in[d^2]\times[d^2]$ with $j<k$ positive semidefinite contractions on $\cH_A$ with $A_1^{jk}+A_2^{jk}\preceq I$, and $B_j$ for $j\in[d^2]$ positive semidefinite contractions on $\cH_B$. In other words, $\rho,A_a^{jk},A_j^\povm,B_j$ determine a quantum model strategy compatible with the scenario in Section \ref{sec:bell}.
Furthermore, given a measurement $X$ on Alice's (or Bob's) side, its compression to $\supp_A\rho$ (or $\supp_B\rho$) is denoted $\hat X$.

\subsection{Measurements in an optimal strategy}
We start by extracting properties of measurements $B_j$ and $A^{jk}_a$ in an optimal strategy for the Bell function \eqref{e:bellfun}.

\begin{prop} \label{prop_relations_measurements_nonSIC}
Assume the Bell function \eqref{e:bellfun} attains $d^2$ at the strategy given by $\rho,A_a^{jk},A_j^\povm,B_j$.
Then
\begin{align}
\label{e:AB_sigurd}
&\sqrt{1-s_{jk}} \big((A^{jk}_1 - A^{jk}_2)\otimes I\big)\rho=\big( I\otimes ( B_j- B_k )\big)\rho \quad\forall j<k,\\
\label{e:ortho_sigurd}
&\big(A^\povm_j\otimes (I - B_j)\big)\rho=0\quad \forall j,
\end{align}
$A^{jk}_1-A^{jk}_2$ preserve $\supp_A\rho$ and $B_j$ preserve $\supp_B\rho$, and the following hold for the compressions of measurements to the local support of $\rho$:
\begin{align}
\label{e:Bproj_sigurd}
&\hat B_j^2 = \hat B_j \quad \forall j, \\
\label{e:Bcomplete_sigurd}
&\sum_j \hat B_j = d I, \\
\label{e:Aproj_sigurd}
&(\hat A^{jk}_1)^2 = \hat A^{jk}_1,\ 
(\hat A^{jk}_2)^2 = \hat A^{jk}_2,\ 
\hat A^{jk}_1\hat A^{jk}_2 = 0 \quad \forall j<k, \\
\label{e:BSIC_sigurd}
&\hat B_j \hat B_k \hat B_j = s_{jk}\hat B_j \quad \forall j \neq k.
\end{align}
\end{prop}

\begin{proof}
By Proposition \ref{p:bellmax} and its proof,
the operator $\Theta_d$ from \eqref{e:SOS_new_finish} satisfies $\tr(\Theta_d\rho)=0$. Furthermore, $\Theta_d\rho=0$ because $\Theta_d$ and $\rho$ are positive semidefinite. Since $\Theta_d$ is a sum of hermitian squares in \eqref{e:SOS_new_finish}, it follows that \eqref{e:AB_sigurd}, \eqref{e:ortho_sigurd} and
\begin{align}
\label{e:frombell2}\big(d I\otimes I - \sum_j I\otimes B_j\big)\rho&=0,\\
\label{e:frombell3}\big(I\otimes (B_j - B_j^2)\big)\rho&=0, \\
\label{e:frombell4}\left(\left(A^{jk}_1 + A^{jk}_2 - ( A^{jk}_1 - A^{jk}_2 )^2\right)\otimes I\right)\rho&=0
\end{align}
hold for all $j< k$.
Next, notice that applying partial traces and Lemma \ref{l:supp} to \eqref{e:AB_sigurd} yields
\begin{equation}\label{e:range1}
(B_j-B_k)\supp_B\rho\subseteq\supp_B\rho,\quad(A^{jk}_1-A^{jk}_2)\supp_A\rho\subseteq\supp_A\rho \qquad\text{for all }j< k.
\end{equation}
In particular, the differences $B_j-B_k$ preserve $\supp_B\rho$ for all $j,k$. 
Thus the same holds for $\sum_k (B_j-B_k)=d^2B_j-\sum_kB_k$, which acts as $d^2B_j-dI$ on $\supp_B\rho$ by \eqref{e:frombell2}. Therefore
\begin{equation}\label{e:range2}
B_j\supp_B\rho\subseteq\supp_B\rho\qquad\text{for all }j.
\end{equation}
Thus when measurements are replaced by their compressions to the local support of $\rho$, the above equations become
\begin{align}
\label{e:frombell6} d I - \sum_j \hat B_j&=0,\\
\label{e:frombell7} \hat B_j - \hat B_j^2&=0, \\
\label{e:frombell8} \hat A^{jk}_1 + \hat A^{jk}_2 - (\hat A^{jk}_1 - \hat A^{jk}_2 )^2&=0
\end{align}
for all $j< k$ (more precisely, \eqref{e:range1} and \eqref{e:range2} are needed because $B_j$ and $A^{jk}_1-A^{jk}_2$ appear nonlinearly in \eqref{e:frombell3} and \eqref{e:frombell4}, respectively). In particular, \eqref{e:Bproj_sigurd} and \eqref{e:Bcomplete_sigurd} hold, and \eqref{e:Aproj_sigurd} follows by \eqref{e:frombell8} and Lemma \ref{l:little}.
We are left with proving \eqref{e:BSIC_sigurd}.
Left-multiplying \eqref{e:AB_sigurd} by $\sqrt{1-s_{jk}}  (A^{jk}_1 - A^{jk}_2 ) \otimes I $ yields
\begin{equation}\label{e:frombell10}
\begin{split}
&(1-s_{jk}) \left(( A^{jk}_1 -  A^{jk}_2)^2 \otimes I\right)\rho 
=  \sqrt{1-s_{jk}} \left(( A^{jk}_1 -  A^{jk}_2) \otimes ( B_j - B_k)\right)\rho \\
 =\,&\big( I \otimes ( B_j - B_k) \big)
\left( \sqrt{1-s_{jk}} ( A^{jk}_1 -  A^{jk}_2) \otimes I\right) \rho
=  \left(I \otimes ( B_j -  B_k)^2\right)\rho.
\end{split}
\end{equation}
We left-multiply \eqref{e:frombell10} by $\sqrt{1-s_jk}( A^{jk}_1 -  A^{jk}_2) \otimes I$ one more time, and note that the projectivity and orthogonality of $\hat A^{jk}_a$ as in \eqref{e:Aproj_sigurd} imply that $(\hat A^{jk}_1 - \hat A^{jk}_2)^3 = \hat A^{jk}_1 - \hat A^{jk}_2$. By \eqref{e:range1} we therefore have
\begin{equation*}
\left(\sqrt{1-s_{jk}}^3 (A^{jk}_1 - A^{jk}_2)\otimes I\right)\rho= \left(I \otimes ( B_j - B_k)^3\right)\rho,
\end{equation*}
which together with \eqref{e:AB_sigurd} yields
\begin{equation}\label{e:cube0}
\left(I\otimes (1-s_{jk}) ( B_j - B_k)\right)\rho= \left(I\otimes (B_j - B_k)^3\right)\rho
\end{equation}
for $j<k$. By \eqref{e:range2} and symmetry we have
\begin{equation}\label{e:cube}
(1-s_{jk}) (\hat B_j - \hat B_k)=(\hat B_j - \hat B_k)^3
\end{equation}
for all $j\neq k$. Then \eqref{e:frombell6}, \eqref{e:frombell7}, \eqref{e:cube} together with Lemma \ref{l:alt} show that \eqref{e:BSIC_sigurd} holds.
\end{proof}

In particular, Proposition \ref{prop_relations_measurements_nonSIC} shows that the compressions of $B_j$ from an optimal strategy determine a representation of the C*-algebra $\cA_S$ from Section \ref{sec:cstar}.
Next, we construct operators $C_j$ on Alice's side whose compressions also determine a representation of $\cA_S$. For $j\in[d^2]$ denote
\begin{equation*}
C_j := \frac{1}{d^2} \left(dI+ 
\sum_{k \neq j}\sqrt{1-s_{jk}}( A^{jk}_1 -  A^{jk}_2 ) \right),
\end{equation*}
where we write $ A_a^{jk}= A_a^{kj}$ for $k<j$.
By Proposition \ref{prop_relations_measurements_nonSIC} we have
\begin{equation}\label{e:CandB}
C_j\supp_A\rho\subseteq\supp_A\rho,\quad (C_j \otimes I)\rho = (I \otimes B_j)\rho \quad \forall j.
\end{equation}
Then \eqref{e:CandB} and \eqref{e:Bproj_sigurd}, \eqref{e:Bcomplete_sigurd}, \eqref{e:BSIC_sigurd} in Proposition \ref{prop_relations_measurements_nonSIC} imply that $\hat C_j$ are projections that add up to $d I$, 
and $\hat C_j \hat C_k \hat C_j = s_{jk} \hat C_j$, so $\hC_j$ indeed determine a representation of $\cA_S$.
We can use these operators to deduce some partial information on the remaining measurements of Alice, that is, the $A^\povm_j$ operators.
Recall that every finite-dimensional representation of the C*-algebra $\cA_S$ is a direct sum of irreducible ones, whose dimensions are multiples of $d$ (Proposition \ref{p:rk}).

\begin{prop}\label{p:Apovm}
Assume the Bell function \eqref{e:bellfun} attains $d^2$ at the strategy given by $\rho,A_a^{jk},A_j^\povm,B_j$.
Choose a basis of $\supp_A\rho$ such that
\begin{equation}\label{e:decomp}
\hC_j=\bigoplus_{\alpha=1}^L I_{e_\alpha}\otimes \hC_{j,\alpha}
\end{equation}
where $(\hC_{j,\alpha})_j$ determine pairwise non-isomorphic irreducible representations of $\cA_S$ of dimension $r_\alpha d$ for $r_\alpha \in\N$.
With respect to the decomposition \eqref{e:decomp}, for every $\alpha\in[L]$ the $(\alpha,\alpha)$-block of $\hA_j^\povm$ in $\mtxc{e_\alpha}\otimes\mtxc{r_\alpha d}$ equals
$$I_{e_\alpha}\otimes \frac{1}{d}
\hC_{j,\alpha}+W_{j,\alpha},$$
where $\tr_{\C^{r_\alpha d}}(W_{j,\alpha})=0\in \mtxc{e_\alpha}$.
\end{prop}

\begin{proof}
Proposition \ref{prop_relations_measurements_nonSIC} implies
\begin{align*}
\left(\big(A^\povm_j(I - C_j)\big) \otimes I\right)\rho
&=\big(A^\povm_j\otimes I\big)\left((I - C_j) \otimes I\right)\rho\\
&=\big(A^\povm_j\otimes I\big)\left(I \otimes (I - B_j)\right)\rho\\
&=\left(A^\povm_j \otimes (I - B_j)\right)\rho=0
\end{align*}
and, since $C_j$ preserves $\supp_A\rho$, one obtains
$\hA^\povm_j(I - \hC_j)=0$.
Since $\hat A^\povm_j$ and $\hat C_j$ are hermitian, we obtain
\begin{equation*}
\hat A^\povm_j\hat C_j=\hat A^\povm_j=\hat C_j\hat A^\povm_j \quad \forall j.
\end{equation*}
Let $(X_{j\alpha pq})_{p,q=1}^{e_\alpha}$ 
denote the $(\alpha,\alpha)$-block of $\hA_j^\povm$ with respect to the decomposition \eqref{e:decomp}. Then
\begin{equation}\label{e:AC}
X_{j\alpha pq}\hC_{j,\alpha}=X_{j\alpha pq}=\hC_{j,\alpha}X_{j\alpha pq}
\end{equation}
for all $j$ and $p,q\in[e_\alpha]$. 
Note that $\hC_{1,\alpha},\dots,\hC_{d^2,\alpha}$ have traces $r_\alpha$ and are linearly independent by Proposition \ref{p:rk}. Therefore
\begin{equation}\label{e:lincomb}
X_{j\alpha pq}=\sum_{k=1}^{d^2} \lambda_{j\alpha pq,k}\hC_{k,\alpha}+W_{j\alpha pq}
\end{equation}
for unique $\lambda_{j\alpha pq,k}\in\C$ and $W_{j\alpha pq}\in\mtxc{r_\alpha d}$ satisfying $\tr(W_{j\alpha pq}\hC_{\ell,\alpha})=0$ for all $\ell$. In particular, $\tr(W_{j\alpha pq})=0$ because the $\hC_{\ell,\alpha}$ add up to a multiple of identity. Then \eqref{e:AC} and \eqref{e:lincomb} together with $\hC_{j,\alpha}\hC_{k,\alpha}\hC_{j,\alpha}=s_{jk}\hC_{j,\alpha}$ imply
\begin{equation}\label{e:linsys}
s_{\ell j}\tr X_{j\alpha pq}
=\tr\left(\hC_{\ell,\alpha}X_{j\alpha pq}\right)=
\sum_{k=1}^{d^2} \lambda_{j\alpha pq,k}\tr\left(\hC_{\ell,\alpha}\hC_{k,\alpha}\right)
=\sum_{k=1}^{d^2} s_{\ell k}r_\alpha \lambda_{j\alpha pq,k}
\end{equation}
for all $\ell\in[d^2]$. Since the matrix $S$ is invertible by Lemma \ref{l:matrix}, the linear system \eqref{e:linsys} of $d^2$ equations in unknowns $\lambda_{j\alpha pq,1},\dots,\lambda_{j\alpha pq,d^2}$ has a unique solution, namely 
$\lambda_{j\alpha pq,j}= \frac{\tr X_{j\alpha pq}}{r_\alpha}$
and $\lambda_{j\alpha pq,k}=0$ for $k\neq j$.
In particular,
\begin{equation}\label{e:mltp}
X_{japq}=\frac{\tr X_{j \alpha pq}}{r_\alpha} \hC_{j,\alpha}+W_{j\alpha pq}\quad \forall j,p,q.
\end{equation}
Since $\hA^\povm_j$ is a POVM, we have $\sum_j X_{j \alpha pp}=I$ and $\sum_j X_{j \alpha pq}=0$ for $p\neq q$. Then \eqref{e:mltp}, linear independence of the $\hC_{j,\alpha}$, 
$\spa \{\hC_{j,\alpha}\}_j\cap \spa \{W_{j\alpha pq}\}_{j,p,q}=\{0\}$ and the relation $\sum_j\frac{1}{d}\hC_{j,\alpha}=I$ altogether
imply
$\frac{\tr X_{j \alpha pp}}{r_\alpha}=\frac1d$ and $\frac{\tr X_{j \alpha pq}}{r_\alpha}=0$ for $p\neq q$.
\end{proof}

\subsection{State factorization}

The preceding characterization of measurements and state block decomposition allow us to prove that $2 \log(d)$ bits of randomness can be extracted from the outcome of the $\povm$ setting of Alice, by showing that a classical-quantum state between Alice and Eve necessarily factors (Theorem \ref{t:main}).

We start with a technical lemma.

\begin{lem}\label{l:todo}
Let $\cH_C,\cH_D,\cH_E$ be finite-dimensional Hilbert spaces, with decompositions $\cH_C=\bigoplus_\alpha \cH_{C_\alpha}$ and $\cH_D=\bigoplus_\alpha \cH_{D_\alpha}$ indexed by a common index set.
Suppose there are pure states $\ket{\psi}\in\cH_C\otimes \cH_D\otimes \cH_E$ and $\ket{\tau_\alpha}\in\cH_{C_\alpha}$ such that $\tr_E(\dyad{\psi})$ is a mixture of pure states in
$$\bigoplus_\alpha \ket{\tau_\alpha}\otimes \cH_{D_\alpha}.$$
Then
$$
\ket{\psi}\in \bigoplus_\alpha \ket{\tau_\alpha}
\otimes\cH_{D_\alpha}\otimes \cH_E.
$$
\end{lem}

\begin{proof}
With respect to the identification
$$\cH_C\otimes \cH_D\otimes \cH_E=\bigoplus_{\alpha,\beta} \cH_{C_\alpha}\otimes \cH_{D_\beta}\otimes \cH_E$$
we write $\ket{\psi}=\bigoplus_{\alpha,\beta} \ket{\psi_{\alpha\beta}}$ for $\ket{\psi_{\alpha\beta}}\in \cH_{C_\alpha}\otimes \cH_{D_\beta}\otimes \cH_E$.
Let us consider the diagonal blocks $\dyad{\psi_{\alpha\beta}}$ of $\dyad{\psi}$ with respect to this decomposition.
If $\alpha\neq\beta$, then $\tr_E(\dyad{\psi_{\alpha\beta}})=0$ by the assumption on $\tr_E(\dyad{\psi})$, and so $\ket{\psi_{\alpha\beta}}=0$. If $\alpha=\beta$, then $\tr_E(\dyad{\psi_{\alpha\alpha}})$ is a mixture of pure states in $\ket{\tau_\alpha}\otimes \cH_{D_\alpha}$. In particular, $\tr_{D_\alpha E}(\dyad{\psi_{\alpha\alpha}})$ is proportional to the pure state $\dyad{\tau_\alpha}$. Therefore $\dyad{\psi_{\alpha\alpha}}$ is a product bipartite state on $\cH_{C_\alpha}\otimes (\cH_{D_\alpha}\otimes \cH_E)$, so $\ket{\psi_{\alpha\alpha}}=\ket{\tau_\alpha}\otimes \ket{\chi_\alpha}$ for some $\ket{\chi_\alpha}\in \cH_{D_\alpha}\otimes \cH_E$.
Thus $\ket{\psi}= \bigoplus_\alpha \ket{\tau_\alpha}\otimes \ket{\chi_\alpha}$.
\end{proof}

A pure tripartite state $\ket{\psi}\in\cH_A\otimes\cH_B\otimes\cH_E$ is a purification of the mixed bipartite state $\rho$ on 
$\cH_A\otimes\cH_B$ if $\rho=\tr_E(\dyad{\psi})$.
The device-independent randomness of 
the $\povm$ setting outcome is bounded from below by the conditional von Neumann entropy $H(A|E)_{\rho_{AE}}$ of the classical-quantum state
\begin{equation}\label{e:rhoAE}
\rho_{AE} = \sum_{j=1}^{d^2} \ketbraq{j}_A 
\otimes \tr_{AB}\left[ \ketbraq{\psi} ( A^\povm_j \otimes I_B \otimes I_E)\right]
\end{equation}
where $\ket{\psi}\in\cH_A\otimes\cH_B\otimes\cH_E$ is the worst-case purification of $\rho$; that is, the purification of $\rho$ that gives the lowest value of $H(A|E)_{\rho_{AE}}$. 

\begin{thm}\label{t:main}
Suppose the state $\rho$ and measurement $(A^\povm_j)_j$ appear in an optimal quantum strategy for the Bell function \eqref{e:bellfun}, and let $\ket{\psi}\in\cH_A\otimes\cH_B\otimes\cH_E$ be a purification of $\rho$. Then
$$
\sum_{j=1}^{d^2} \ketbraq{j}_A 
\otimes \tr_{AB}\left[ \ketbraq{\psi} ( A^\povm_j \otimes I_B \otimes I_E)\right]=\left(\frac{1}{d^2}\sum_{j=1}^{d^2}\ketbraq{j}\right)\otimes
\sigma_E$$
for some state $\sigma_E$ on $\cH_E$.

In particular, the maximal violation of the Bell inequality \eqref{e:bellfun} certifies $2 \log(d)$ bits of device-independent randomness from the outcome of the $\povm$ setting of Alice.
\end{thm}

\begin{proof}
The operators $\hC_j$ and $\hB_j$ determine finite-dimensional representations of $\cA_S$, which are direct sums of irreducible ones whose dimensions are multiples of $d$ by Proposition \ref{p:rk}.
By \eqref{e:CandB} and Proposition \ref{p:maxent} there are isometries $U:\C^{D_A d}\to\cH_A$ and $V:\C^{D_B d}\to\cH_B$ with $\ran U=\supp_A\rho$ and $\ran V=\supp_B\rho$ such that
\begin{equation}\label{e:blocks}
U^*C_jU\in\bigoplus_\alpha (I_{e_\alpha}\otimes\mtxc{r_\alpha d}),
\quad
V^*B_jV\in\bigoplus_\alpha (I_{f_\alpha}\otimes\mtxc{r_\alpha d}),
\end{equation}
where $D_A=\sum_\alpha e_\alpha r_\alpha$ and $D_B=\sum_\alpha f_\alpha r_\alpha$. Furthermore, 
$(U\otimes V)^*\rho(U\otimes V)$ is a mixture of pure states in
$$
\bigoplus_\alpha \left(\C^{e_\alpha}\otimes \C^{f_\alpha}\right)\otimes\ket{\varphi_{r_\alpha d}}.
$$
With respect to the decomposition \eqref{e:blocks} let
$$U^*C_jU=\bigoplus_\alpha I_{e_\alpha}\otimes \hC_{j,\alpha}$$
where $\tr \hC_{j,\alpha}=r_\alpha$.
By Proposition \ref{p:Apovm},
\begin{equation}\label{e:apovm}
U^*A^\povm_jU = \cN_j+\bigoplus_\alpha\left(
I_{e_\alpha}\otimes \frac{1}{d}\hC_{j,\alpha}
+W_{j,\alpha}\right),
\end{equation}
where $\tr_{\C^{r_\alpha d}}(W_{j,\alpha})=0$, and all diagonal $(\alpha,\alpha)$-blocks of $\cN_j\in \mtxc{D_A d}$ are zero.

Let $\ket{\psi}$ be a purification of $\rho$. Then $\rho=\tr_E( \ketbraq{\psi} )$ and so
$(U\otimes V)^*\rho(U\otimes V)=\tr_E((U\otimes V\otimes I_E)^* \ketbraq{\psi} (U\otimes V\otimes I_E))$.
By Lemma \ref{l:todo}, up to a suitable shuffle of direct sums and tensor products we have
\begin{equation*}
\ket{\hat{\psi}}:=(U^*\otimes V^* \otimes I_E) \ket{\psi}=\bigoplus_\alpha 
\ket{\chi_\alpha}\otimes\ket{\varphi_{r_\alpha d}}
\end{equation*}
for some $\ket{\chi_\alpha}\in (\C^{e_\alpha}\otimes \C^{f_\alpha})\otimes \cH_E$.
Let us record two observations on the block interactions of the decomposition \eqref{e:blocks}.
In the following calculations, equalities are valid up to a compatible shuffle of tensor products.
Firstly, for all $\alpha$ and $j$ we have
\begin{equation}\label{e:hiccup1}
\begin{split}
&\tr_{AB}\left[
\left(\dyad{\varphi_{r_\alpha d}}\otimes\dyad{\chi_\alpha}\right)
\left(W_{j,\alpha}\otimes I_{f_\alpha r_\alpha d} \otimes I_E \right)
\right] \\
=\,&
\tr_{\C^{e_\alpha}\otimes \C^{f_\alpha}}
\left[\tr_{\C^{r_\alpha d}\otimes \C^{r_\alpha d}}\left[
\left(\dyad{\varphi_{r_\alpha d}}\otimes\dyad{\chi_\alpha}\right)
\left(W_{j,\alpha}\otimes I_{f_\alpha r_\alpha d} \otimes I_E \right)
\right]\right] \\
=\,&
\tr_{\C^{e_\alpha}\otimes \C^{f_\alpha}}
\left[\dyad{\chi_\alpha} \left(\frac{1}{r_\alpha d}\tr_{\C^{r_\alpha d}}\left(W_{j,\alpha}\right)\otimes I_{f_\alpha} \otimes I_E \right)
\right]
=0
\end{split}
\end{equation}
because $\tr_{\C^{r_\alpha d}}(W_{j,\alpha})=0$.
Secondly, for all $j$ we have
\begin{equation}\label{e:hiccup2}
\begin{split}
&\tr_{AB}\left[
\dyad{\hat{\psi}}
\left( \cN_j \otimes I_{D_Bd} \otimes I_E \right)
\right] \\
=\,&
\tr_{AB}\left[
\left(\bigoplus_{\alpha,\beta} 
\ketbra{\varphi_{r_\alpha d}}{\varphi_{r_\beta d}}\otimes \ketbra{\chi_\alpha}{\chi_\beta}
\right)
\left( \cN_j \otimes I_{D_Bd} \otimes I_E \right)
\right] \\
=\,&
\tr_{AB}\left[
\Big(
\sum_\gamma
\left(\ketbra{\varphi_{r_\alpha d}}{\varphi_{r_\gamma d}}\otimes \ketbra{\chi_\alpha}{\chi_\gamma}\right)
\left( (\cN_j)_{\gamma,\beta} \otimes (I_{D_Bd})_{\gamma,\beta} \otimes I_E \right)
\Big)_{\alpha,\beta}
\right] \\
=\,&0
\end{split}
\end{equation}
because $(\cN_j)_{\alpha,\alpha}=0$ for all $\alpha$, and $(I_{D_Bd})_{\alpha,\beta}=0$ for all $\alpha\neq \beta$.
For the sake of readability, $I_B$ denotes the identity on Bob's system or its subsystems (depending on the context) in the following calculations. 
For every $j\in[d^2]$, we can use \eqref{e:apovm}, \eqref{e:hiccup1} and \eqref{e:hiccup2} together with $(UU^* \otimes VV^* \otimes I_E)\ket{\psi} = \ket{\psi}$ from 
Proposition \ref{p:maxent} (and Lemma \ref{l:supp}(i))
to calculate
\begin{align*}
& \tr_{AB}\left[ \ketbraq{\psi}( A^\povm_j \otimes I_B \otimes I_E )\right] \\
=\,& \tr_{AB}\left[
(U^*\otimes V^* \otimes I_E) \ketbraq{\psi} (U\otimes V \otimes I_E )( U^*A^\povm_jU \otimes I_B \otimes I_E )
\right]\\
=\,& \tr_{AB}\left[
\dyad{\hat{\psi}}
\left( \left(\cN_j+\bigoplus_\alpha\left(
I_{e_\alpha}\otimes \frac{1}{d}\hC_{j,\alpha}
+W_{j,\alpha}\right)\right)
\otimes I_B \otimes I_E \right)
\right]\\
=\,& \tr_{AB}\left[
\dyad{\hat{\psi}}
\left(\left(\bigoplus_\alpha
I_{e_\alpha}\otimes \frac{1}{d}\hC_{j,\alpha}
\right)
\otimes I_B \otimes I_E \right)
\right]\\
=\,& \sum_\alpha\tr_{AB}\left[
\left(\dyad{\chi_\alpha}\otimes \dyad{\varphi_{r_\alpha d}}\right)
\left( \left(I_{e_\alpha}\otimes 
\frac{\hC_{j,\alpha}}{d}\right)
\otimes I_B \otimes I_E \right)
\right]\\
=\,& \sum_\alpha 
\tr\left(
\dyad{\varphi_{r_\alpha d}}\left(
\frac{\hC_{j,\alpha}}{d} \otimes I_{r_\alpha d}\right)
\right)\tr_{\C^{e_\alpha}\otimes\C^{f_\alpha}}(\dyad{\chi_\alpha})
\\
=\,& \sum_\alpha \frac{1}{r_\alpha d^2}\tr\left(\hC_{j,\alpha}\right)\tr_{\C^{e_\alpha}\otimes\C^{f_\alpha}}(\dyad{\chi_\alpha})
\\
=\,& 
\sum_\alpha \frac{1}{d^2}\tr_{\C^{e_\alpha}\otimes\C^{f_\alpha}}(\dyad{\chi_\alpha}),
\end{align*}
which is independent of $j$.
Thus the state $\rho_{AE}$ in \eqref{e:rhoAE} can be written as
\begin{align*}
\rho_{AE}
&=\sum_j \ketbraq{j} \otimes \tr_{AB}\left[ \ketbraq{\psi}( A^\povm_j \otimes I_B \otimes I_E )\right] \\
&=\frac{1}{d^2}\sum_j\ketbraq{j}\otimes
\sigma_E
\end{align*}
where $\sigma_E := \sum_\alpha \tr_{\C^{e_\alpha}\otimes\C^{f_\alpha}}(\dyad{\chi_\alpha})$ is a state on $\cH_E$. It therefore follows that
\begin{align*}
H(A|E)_{\rho_{AE}} & \left. = H(AE)_{\rho_{AE}} - H(E)_{\rho_{AE}} = H\left( \frac{1}{d^2} \sum_{j=1}^{d^2} \ketbraq{j} \otimes \sigma_E \right) - H( \sigma_E ) \right. \\
& \left. = H\left( \frac{1}{d^2} \sum_{j=1}^{d^2} \ketbraq{j} \right) + H( \sigma_E ) - H( \sigma_E ) = 2 \log(d)
\right.
\end{align*}
for any purification $\ket{\psi}$ compatible with the observed correlation.
\end{proof}

\section{Classical value of the BIC-POVM Bell function}
\label{app:QCgap}

Let $d\ge 2$, and let $S\in\mtxr{d^2}$ be a matrix induced by a BIC-POVM. 
Consider the Bell function
\begin{equation}\label{e:bellfunA}
\begin{split}
&2\sum_{j<k}\sqrt{1-s_{jk}} \big[ p(1,1|(j,k),j) + p(2,1|(j,k),k) - p(1,1|(j,k),k) - p(2,1|(j,k),j) \big] \\
&- \sum_{j<k}(1-s_{jk}) \big[ p_A(1|(j,k)) + p_A(2|(j,k)) \big] - d(d-2) \sum_{j=1}^{d^2} p_B(1|j) - \sum_{j=1}^{d^2} p(j,\perp|\povm,j)
\end{split}
\end{equation}
introduced in Section \ref{sec:bell}. By Proposition \ref{p:bellmax}, its maximal quantum value equals $d^2$. In this section we provide an expression and an upper bound for its maximal classical value.

\begin{prop}\label{p:bellmaxclas}
Let $S\in\mtxr{d^2}$ be a matrix induced by a BIC-POVM. Then the maximal classical value of the Bell function \eqref{e:bellfunA} equals
\begin{equation}\label{e:bellmaxclas}
\max_{\substack{J\subseteq [d^2],\\ 0<|J|<2d}}
\left(-d(d-2)|J|
+\sum_{j\in J,k\notin J}\left(2\sqrt{1-s_{jk}}-(1-s_{jk})\right)
\right)
\end{equation}
In particular, \eqref{e:bellfunA} is bounded from above by 
$$d^2-\frac14
\left(\min_{\substack{J\subseteq [d^2],\\ 0<|J|<2d}}
\sum_{j\in J,k\notin J}s_{jk}^2\right)
<d^2.$$
\end{prop}

\begin{proof}
Since \eqref{e:bellfunA} is a convex optimization problem, its solution is attained at a deterministic strategy.
Thus it suffices to solve the optimization problem
\begin{equation}\label{e:opt0}
\begin{split}
\max\quad & 2\sum_{j<k}\sqrt{1-s_{jk}}(a^{jk}_1-a^{jk}_2)(b_j-b_k)
-\sum_{j<k}(1-s_{jk})(a^{jk}_1+a^{jk}_2) \\
&-d(d-2)\sum_j b_j-\sum_j a^\povm_j(1-b_j) \\
\text{subject to:}\quad & a^{jk}_1,a^{jk}_2,b_j,a^\povm_j\in\{0,1\},
\ a^{jk}_1a^{jk}_2=0, \text{ exactly one of }a_j^\povm\text{ is nonzero}.
\end{split}
\end{equation}
It is easier to analyze \eqref{e:opt0} with auxiliary notation $a^{jk}_i=a^{kj}_i$ for $j> k$ that makes \eqref{e:opt0} equivalent to
\begin{equation}\label{e:opt}
\begin{split}
\max\quad & \sum_{j,k}\left(\sqrt{1-s_{jk}}(a^{jk}_1-a^{jk}_2)(b_j-b_k)
-\frac{1-s_{jk}}{2}(a^{jk}_1+a^{jk}_2)\right)\\
&-d(d-2)\sum_j b_j-\sum_j a^\povm_j(1-b_j) \\
\text{subject to:}\quad & 
a^{jk}_1=a^{kj}_1,a^{jk}_2=a^{kj}_2,b_j\in\{0,1\},
\ a^{jk}_1a^{jk}_2=0, \text{ exactly one of }a_j^\povm\text{ is nonzero}.
\end{split}
\end{equation}
First, observe that $2\sqrt{1-t}-(1-t)$ is a monotone decreasing function $(0,1)\to (0,1)$. Let $J\subseteq[d^2]$ be arbitrary. Suppose we fix the $b_j$ arguments in the objective function of \eqref{e:opt} as $b_j=1$ if $j\in J$ and $b_j=0$ if $j\notin J$. 
To maximize the value of this partially evaluated objective function, it is then necessary to set:
\begin{itemize}
    \item $a_1^{jk}=1$ if $j\in J$ and $k\notin J$, 
    \item $a_2^{jk}=1$ if $j\notin J$ and $k\in J$, 
    \item $a_1^{jk}=a_2^{jk}=0$ if $j,k\in J$ or $j,k\notin J$, 
    \item $a_j^\povm=1$ for some $j\in J$, unless $J=\emptyset$ in which case the choice of $j$ is irrelevant.
\end{itemize}
For these arguments, the objective function of \eqref{e:opt} evaluates as
\begin{equation}\label{vJ}
v(J):=-d(d-2)|J|
+\sum_{j\in J,k\notin J}\left(2\sqrt{1-s_{jk}}-(1-s_{jk})\right)
\end{equation}
if $J\neq\emptyset$, and as $-1$ if $J=\emptyset$.
By applying the estimate
$$1-t<2\sqrt{1-t}-(1-t)<1-\frac{t^2}{4} \qquad \text{for } t\in (0,1)$$ 
in \eqref{vJ} we see that
$$
v(J) <-d(d-2)|J|
+\sum_{j\in J,k\notin J}\left(1-\tfrac{1}{4}s_{jk}^2\right)
=|J|(2d-|J|)-\tfrac14 \sum_{j\in J,k\notin J}s_{jk}^2.
$$
In particular, $v(J)<0$ for $\abs{J}=0$ and $\abs{J}\geq 2d$. 
On the other hand, for $J=\{j\}$ one has
$$v(\{j\})=-d(d-2)
+\sum_{k\neq j}\left(2\sqrt{1-s_{jk}}-(1-s_{jk})\right)
> -d(d-2)
+\sum_{k\neq j}(1-s_{jk})=d.
$$
Therefore the classical value equals $\max_{0<|J|<2d}v(J)$. Moreover, the expression $|J|(2d-|J|)$ is at most $d^2$, leading to the following upper bound on the classical value,
\begin{equation*}
    d^2-\tfrac{1}{4}\min_{0<|J|<2d}\sum_{j\in J,k\notin J}s_{jk}^2.
\end{equation*}
Note that this value is strictly less than $d^2$ because $\sum_{j\in J,k\notin J}s_{jk}^2>0$ for every $0<|J|<d^2$ by Lemma \ref{l:matrix}.
\end{proof}

\begin{exa}
Let us give a complete study of the classical value in the case of BIC-POVMs on $\C^2$. A routine calculation shows that up to unitary similarity, every quadruple of projections adding to $2I$ is of the form
\begin{align*}
P_1&=\begin{pmatrix}
1&0\\0&0
\end{pmatrix},\quad 
P_2=\begin{pmatrix}
1-t_1-t_2&-y_1-y_2\\-y_1-y_2&t_1+t_2
\end{pmatrix},\\
P_3&=\begin{pmatrix}
t_1&y_1+iz\\y_1-iz&1-t_1
\end{pmatrix},\quad 
P_4=\begin{pmatrix}
t_2&y_2-iz\\y_2+iz&1-t_2
\end{pmatrix},\\
y_j&=\frac{t_j\sqrt{1-t_1-t_2}}{\sqrt{t_1+t_2}},\quad 
z=\pm\frac{\sqrt{t_1t_2}}{\sqrt{t_1+t_2}}
\end{align*}
where $t_1,t_2\ge0$ and $t_1+t_2\le 1$.
Let $S=(\tr(P_jP_k))_{i,j=1}^4$. Then
$$S=\begin{pmatrix}
1 & 1-t_1-t_2 & t_1 & t_2 \\
1-t_1-t_2 & 1 & t_2 & t_1 \\
t_1 & t_2 & 1 & 1-t_1-t_2 \\
t_2 & t_1 & 1-t_1-t_2 & 1 \\
\end{pmatrix}$$
is invertible for $t_1,t_2,1-t_1-t_2\neq0$. Therefore BIC-POVMs on $\C^2$ are (up to unitary similarity) given by $(\frac12P_j)_j$ as above for parameters $t_1,t_2$ satisfying $t_1,t_2>0$ and $t_1+t_2<1$.
Note that $t_1=t_2=\frac13$ yields a SIC-POVM.
A case-by-case analysis of the formula in Proposition \ref{p:bellmaxclas} shows that the classical game value of the Bell function \eqref{e:bellfunA} associated with the BIC-POVM given by $(t_1,t_2)$ equals
$$
v(t_1,t_2)=2\cdot\left\{
\begin{array}{lr}
t_1+t_2+2(\sqrt{1-t_1}+\sqrt{1-t_2}-1) & 0< t_2\le \frac{1-t_1}{2},\ 0< t_1\le \frac{1-t_2}{2};\\
2(\sqrt{1-t_1}+\sqrt{t_1+t_2})-1-t_2 & 0<t_1\le t_2,\ t_2\ge \frac{1-t_1}{2},\ t_1+t_2<1;\\
2(\sqrt{1-t_2}+\sqrt{t_1+t_2})-1-t_1 & 0<t_2\le t_1,\ t_1\ge \frac{1-t_2}{2},\ t_1+t_2<1.\\
\end{array}
\right.
$$
\begin{figure}[!ht]
\centering
\includegraphics[width=2 in]{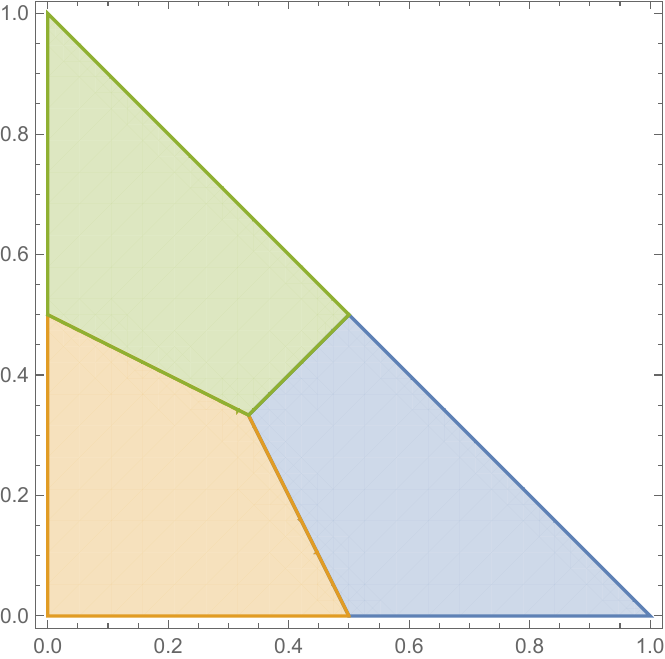}
\includegraphics[width=3 in]{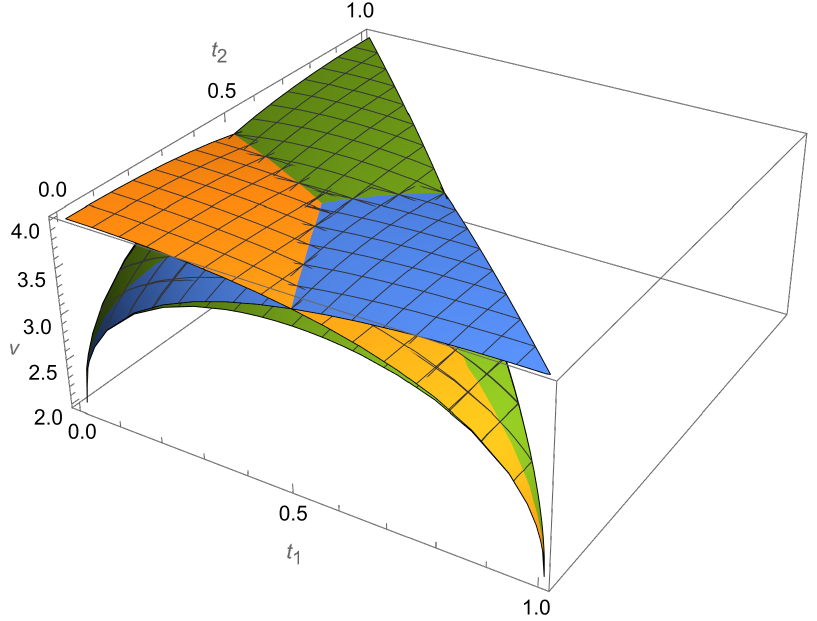}
\caption{The regions and branches of the classical value function $v(t_1,t_2)$.}
\end{figure}

In particular, the formula for $v(t_1,t_2)$ leads to the following observations.
\begin{enumerate}[(i)]
    \item $\lim_{t\to 0}v(t,t)=4$, so the classical value can gets arbitrarily close to the quantum value;
    \item $\lim_{t\to\frac12}v(t,t)=1+2\sqrt{2}<\frac83(\sqrt{6}-1)=v(\frac13,\frac13)$, so a SIC-POVM does not give the largest gap between quantum and classical value;
    \item $1+2\sqrt{2}=\inf\{v(t_1,t_2)\colon t_1,t_2>0,t_1+t_2<1\}$, so the gap between quantum and classical value is at most $3-2\sqrt{2}\approx 0.172$.
\end{enumerate}
\end{exa}

\section{Further remarks on BIC-POVMs}\label{a:icpovm} 

This appendix collects examples and statements which, while not required for the derivation of the main results, are relevant to the broader theme of this paper.

\subsection{BIC-POVMs versus rank-one IC-POVMS}

By definition, a BIC-POVM on $\C^d$ is a $d^2$-outcome IC-POVM of rank-one matrices with trace $\frac1d$.
The $\Z_d\times\Z_d$ covariant IC-POVMs from Section \ref{sec:icpovm} are automatically BIC-POVMs, and moreover consists of pairwise non-orthogonal matrices. However, this is not the case for general rank-one IC-POVMs.

\begin{exa}
The quadruple
$$
\begin{pmatrix}
\frac12&0\\0&0
\end{pmatrix},
\begin{pmatrix}
\frac18&\frac{-i}{2\sqrt{6}}\\ \frac{i}{2\sqrt{6}}&\frac13
\end{pmatrix},
\begin{pmatrix}
\frac18&\frac{-1}{2\sqrt{6}}\\ \frac{-1}{2\sqrt{6}}&\frac13
\end{pmatrix},
\begin{pmatrix}
\frac14&\frac{1+i}{2\sqrt{6}}\\ \frac{1-i}{2\sqrt{6}}&\frac13
\end{pmatrix}
$$
is a rank-one IC-POVM of pairwise non-orthogonal matrices, but not a BIC-POVM (since not all traces are the same).
\end{exa}

\begin{exa}
Consider the states
$\ket{\psi_1},\dots,\ket{\psi_9}\in\C^3$ given as
\begin{align*}
&\begin{pmatrix}
1\\ 0\\ 0
\end{pmatrix},
\begin{pmatrix}
0\\ 1\\ 0
\end{pmatrix},
\begin{pmatrix}
\sqrt{\frac27}\\ \sqrt{\frac27}\\ \sqrt{\frac37}
\end{pmatrix},
\begin{pmatrix}
-\sqrt{\frac27}\\ \sqrt{\frac27}\\ \sqrt{\frac37}
\end{pmatrix},
\begin{pmatrix}
e^{2it_0}\sqrt{\frac27}\\ e^{2it_1}\sqrt{\frac27}\\ \sqrt{\frac37}
\end{pmatrix},
\\
&\begin{pmatrix}
e^{-2it_0}\sqrt{\frac27}\\ \sqrt{\frac27}\\ \sqrt{\frac37}
\end{pmatrix},
\begin{pmatrix}
-\sqrt{\frac27}\\ e^{2it_1}\sqrt{\frac27}\\ \sqrt{\frac37}
\end{pmatrix},
\begin{pmatrix}
\sqrt{\frac27}\\ e^{2it_2}\sqrt{\frac27}\\ \sqrt{\frac37}
\end{pmatrix},
\begin{pmatrix}
\sqrt{\frac27}\\ e^{2it_3}\sqrt{\frac27}\\ \sqrt{\frac37}
\end{pmatrix}
\end{align*}
for
$$t_0=\tfrac{\pi}{3},\ t_1=\arctan(\tfrac{7}{\sqrt{3}}),\ 
t_2=\arctan(\sqrt{3}(14+\sqrt{217})),\ 
t_3=\arctan(\sqrt{3}(14-\sqrt{217})).$$
Then $M_j=\frac13 \dyad{\psi_j}$ form a BIC-POVM on $\C^3$ that does not arise from the construction in Section \ref{sec:icpovm} since $M_1M_2=0$.
\end{exa}

Furthermore, $d^2$-outcome rank-one IC-POVMs on $\C^d$ without additional restrictions can be constructed in a very haphazard way. For example, if $\ket{\psi_1},\dots,\ket{\psi_{d^2}}\in\C^d$ are sufficiently generic, then $\dyad{\psi_1},\dots,\dyad{\psi_{d^2}}$ form a basis of $\mtxc{d}$. Their sum is positive definite, and thus factors as $K^2$ for a positive definite $K\in\mtxc{d}$. Then $K^{-1}\dyad{\psi_1}K^{-1},
\dots,K^{-1}\dyad{\psi_{d^2}}K^{-1}$ is a rank-one IC-POVM.

Some more details are needed for a generic construction of BIC-POVMs. First observe that BIC-POVMs on $\C^d$ (up to unitary similarity) are in one-to-one correspondence with hermitian matrices $G\in\mtxc{d^2}$ such that $G_{jj}=1$ for all $j$, the Schur product of $G$ and $\overline{G}$ is invertible, and $\frac1d G$ is a projection. Concretely, $(\frac1d \dyad{\psi_j})_j$ is a BIC-POVM if and only if the matrix $G=(\braket{\psi_j}{\psi_k})_{j,k}$ satisfies the above properties. Thus one can construct BIC-POVMs as follows. Start with a full-rank $d^2\times d$ complex matrix $K_0$. After column orthonormalization of $K_0$ we obtain $K_1$. Then $K_2=K_1K_1^*$ is a projection of rank $d$. Then there exists an effectively computable unitary $U\in\mtxc{d^2}$ such that $K_3=UK_2U^*$ has uniform diagonal entries $\frac{\tr K_2}{d^2}=\frac1d$. Let $G=dK_3$; if $K_0$ was sufficienly generic, the Schur product of $G$ and $\overline{G}$ is invertible. Therefore $G$ has the desired properties, and one can extract a BIC-POVM out of $G$ using unitary diagonalization.

\subsection{BIC-POVM C*-algebra}

Let $S\in\mtxr{d^2}$ be a matrix induced by a BIC-POVM. 
For the sake of simplicity let us furthermore assume that all the entries of $S$ are nonzero (which is for example true for BIC-POVMs from Section \ref{sec:icpovm}).
In Section \ref{sec:cstar} we introduced the C*-algebra $\cA_S$ whose $d$-dimensional irreducible representations correspond to BIC-POVMs that induce S (Proposition \ref{p:lowdim}). However, $\cA_S$ might have other irreducible representations when $d>2$; see Example \ref{ex:counterex} below.
This motivates the introduction of the universal C*-algebra
\begin{alignat*}{2}
\cB_S={\rm C}^*\bigg\langle x_1,\dots,x_{d^2}\colon 
&& x_j=x_j^*=x_j^2\ \forall j,\ \sum_{j=1}^{d^2}x_j=d,\ 
x_jx_kx_j=s_{jk}x_j\ \forall j, k,&\\
&& \big[x_1x_{j_1}x_{j_2}x_1,\,x_1x_{j_3}x_{j_4}x_1\big]=0\ \forall j_1,\dots,j_4&
\bigg\rangle.
\end{alignat*}
Note that $\cB_S$ is a quotient of $\cA_S$.
A straightforward inspection shows that $\cB_S=\cA_S$ for $d\le 2$; however, $\cB_S\neq\cA_S$ in general (see Example \ref{ex:counterex} below).

\begin{lem}\label{l:commutative}
The C*-subalgebra $x_1\cdot \cB_S\cdot x_1$ is abelian.
\end{lem}

\begin{proof}
Observe that $B=\{x_1x_jx_kx_1\colon j,k\}$ generates $x_1\cdot \cB_S\cdot x_1$ as a C*-algebra. Indeed, this follows by induction using
$x_1ux_jvx_1 =\frac{1}{s_{j1}}x_1ux_jx_1\cdot x_1x_jvx_1$.
Therefore $\cB_S$ is abelian since $B$ is a commuting family and $B=B^*$.
\end{proof}

\begin{prop}\label{p:irrB}
Irreducible representations of $\cB_S$ correspond to BIC-POVMs that induce $S$.
\end{prop}

\begin{proof}
First observe that if $\pi$ is a $d$-dimensional representation of $\cA_S$, then $\pi(x_i)$ have rank one, so $\pi(x_j)M\pi(x_j)$ is a scalar multiple of $\pi(x_j)$ for every $M\in\mtxc{d}$ and $j$. Therefore $\pi$ is also a representation of $\cB_S$.
By Propositions \ref{p:rk} and \ref{p:lowdim} it thus suffices to show that every representation of $\cB_S$ has a sub-representation of dimension $d$.
Let $\pi:\cB_S\to B(\cH)$ be a representation of $\cB_S$ on a (nonzero) Hilbert space $\cH$, and denote $X_j=\pi(x_j)$. Since the $X_j$ add up to a multiple of the identity operator, at least one of them is nonzero; since $X_jX_1X_j=s_{j1}X_j$ for all $j$, we in particular have $X_1\neq0$.
Let $\cC=x_1\cdot \cB_S\cdot x_1$. By Lemma \ref{l:commutative}, the C*-algebra $\cC$ is abelian, and nonzero since $X_1\neq0$. Therefore the restriction of $\pi$ to $\cC$ has a one-dimensional sub-representation. That is, there exists a unit vector $\ket{\psi}\in\cH$ that is an eigenvector for every element of $\pi(\cC)$; namely, for every tuple $(j_1,\dots,j_m)\in\{1,\dots,d^2\}^m$ there is $\lambda_{j_1\dots j_m}\in\C$ such that
$$X_1X_{j_1}\cdots X_{j_m}X_1\ket{\psi}=\lambda_{j_1\dots j_m}\ket{\psi}.$$
In particular, $\ket{\psi}$ lies in the range of $X_1$.
Let $\cK\subset\cH$ be the span of 
$\{X_1\ket{\psi},\cdots,X_{d^2}\ket{\psi}\}$.

Firstly, we claim that $\cK$ is an invariant subspace for $X_1,\dots,X_{d^2}$, and therefore gives rise to a finite-dimensional sub-representation of $\pi$.
Observe that $N\in B(\cH)$ is zero if and only if $X_1X_jN=0$ for $j=1,\dots,d^2$ (because $X_jX_1X_j=s_{j1}X_j$, and $X_j$ add up to a nonzero multiple of the identity). Let $k,\ell\in\{1,\dots,d^2\}$ be arbitrary; we will show that
\begin{equation}\label{e:xkxl}
dX_kX_\ell \ket{\psi}=\sum_{j=1}^{d^2}
\frac{\lambda_{jk\ell}}{s_{1j}} X_j\ket{\psi}.
\end{equation}
By the preceding observation, \eqref{e:xkxl} is equivalent to
\begin{equation}\label{e:xixkxl}
dX_1X_iX_kX_\ell \ket{\psi}=\sum_{j=1}^{d^2}
\frac{\lambda_{jk\ell}}{s_{1j}} X_1X_iX_j\ket{\psi}
\qquad \text{for all }i=1,\dots,d^2.
\end{equation}
The choice of $\ket{\psi}=X_1\ket{\psi}$ and the defining relations for $X_j$ imply
\begin{align*}
\sum_{j=1}^{d^2}\frac{\lambda_{jk\ell}}{s_{1j}} X_1X_iX_j\ket{\psi}
&=\sum_{j=1}^{d^2} \frac{1}{s_{1j}}X_1X_iX_jX_1X_jX_kX_\ell X_1\ket{\psi}
=\sum_{j=1}^{d^2} X_1X_iX_jX_kX_\ell X_1\ket{\psi} \\
&=X_1X_i\left(\sum_{j=1}^{d^2} X_j\right)X_kX_\ell X_1\ket{\psi}
=dX_1X_iX_kX_\ell \ket{\psi}.
\end{align*}
Therefore \eqref{e:xixkxl} holds and consequently \eqref{e:xkxl} holds, so $\cK$ is an invariant subspace of $X_1,\dots,X_{d^2}$.

Secondly, we claim that $\dim\cK=d$. Let $G=(\bra{\psi}X_iX_j\ket{\psi})_{i,j}=(\lambda_{ij})_{i,j} \in\mtxc{d^2}$ be the Gram matrix of the spanning set of $\cK$, and let $D\in\mtxr{d^2}$ be the diagonal matrix whose $k$\textsuperscript{th} diagonal entry equals $\frac{1}{s_{1k}}$. Note that $\dim\cK=\rk G$, and $D$ is positive definite.
Observe that $GDG=dG$. Indeed, the $(i,j)$-entry of $GDG$ equals
\begin{align*}
\sum_{k=1}^{d^2}\frac{\lambda_{ik}\lambda_{kj}}{s_{1k}}
&=\sum_{k=1}^{d^2}\frac{1}{s_{1k}}
\bra{\psi}X_1X_iX_kX_1X_kX_jX_1\ket{\psi}
=\sum_{k=1}^{d^2}
\bra{\psi}X_1X_iX_kX_jX_1\ket{\psi} \\
&=\bra{\psi}X_1X_i\left(\sum_{k=1}^{d^2} X_k\right)X_jX_1\ket{\psi}
=\bra{\psi}X_1X_iX_jX_1\ket{\psi}=d\lambda_{ij}.
\end{align*}
Therefore $\frac1d \sqrt{D}G\sqrt{D}$ is a projection, so
\begin{align*}
\dim\cK&=\rk G=\rk \left(\frac1d \sqrt{D}G\sqrt{D}\right) 
=\tr \left(\frac1d \sqrt{D}G\sqrt{D}\right)
=\frac1d \tr (GD)\\
&=\frac1d\sum_{j=1}^{d^2}\frac{\lambda_{jj}}{s_{1j}}
=\frac1d\sum_{j=1}^{d^2}\frac{1}{s_{1j}}\bra{\psi}X_1X_jX_1\ket{\psi}
=\frac1d \left(\sum_{j=1}^{d^2}\bra{\psi}X_1\ket{\psi}\right)=d,
\end{align*}
as desired.
\end{proof}

Let $S$ be a $d^2\times d^2$ matrix whose diagonal entries equal 1 and off-diagonal entries equal $\frac{1}{d+1}$.
In view of Proposition \ref{p:irrB}, irreducible representations of $\cB_S$ are in one-to-one correspondence with SIC-POVMs (symmetric IC-POVMs).
An analog of Proposition \ref{p:irrB} pertaining to mutually unbiased bases is given in \cite{npa2012,gribling22}.
On the other hand, $\cA_S$ may have representations that do not arise from SIC-POVMs, as it is shown by the following example.

\begin{exa}\label{ex:counterex}
Let $d=3$ and let $S$ be the $9\times 9$ matrix with 1 on the diagonal and $\frac{1}{4}$ elsewhere. In other words, $S$ is induced by a SIC-POVM on $\C^3$. We will show $\cA_S\neq \cB_S$ by producing a 6-dimensional irreducible representation of $\cA_S$ (on the other hand, all irreducible representations of $\cB_S$ are 3-dimensional by Proposition \ref{p:irrB}).
Let $\xi=e^{\frac{\pi i}{12}}$ be the principal 24\textsuperscript{th} root of unity, and consider nine $6\times 6$ matrices $X_1,\dots,X_9$
\begin{align*}
& I_2\oplus 0_4,
\begin{pmatrix}
 \frac{1}{4} & 0 & \frac{-\sqrt{3}}{4} & 0 & 0 & 0 \\
 0 & \frac{1}{4} & 0 & \frac{-\sqrt{3}}{4} & 0 & 0 \\
 \frac{-\sqrt{3}}{4} & 0 & \frac{3}{4} & 0 & 0 & 0 \\
 0 & \frac{-\sqrt{3}}{4} & 0 & \frac{3}{4} & 0 & 0 \\
 0 & 0 & 0 & 0 & 0 & 0 \\
 0 & 0 & 0 & 0 & 0 & 0 
\end{pmatrix},
\begin{pmatrix}
 \frac{1}{4} & 0 & \frac{\sqrt{3}}{4} & 0 & 0 & 0 \\
 0 & \frac{1}{4} & 0 & \frac{\sqrt{3}}{4} & 0 & 0 \\
 \frac{\sqrt{3}}{4} & 0 & \frac{3}{4} & 0 & 0 & 0 \\
 0 & \frac{\sqrt{3}}{4} & 0 & \frac{3}{4} & 0 & 0 \\
 0 & 0 & 0 & 0 & 0 & 0 \\
 0 & 0 & 0 & 0 & 0 & 0 
\end{pmatrix}, \\ 
&
\begin{pmatrix}
 \frac{1}{4} & 0 & \frac{\xi^6}{4} & 0 & \frac{1}{\sqrt{8} \xi^5} & 0 \\
 0 & \frac{1}{4} & 0 & \frac{\xi^6}{4} & 0 & \frac{-1}{\sqrt{8} \xi^5} \\
 \frac{1}{4\xi^6} & 0 & \frac{1}{4} & 0 & \frac{-\xi}{\sqrt{8}} & 0 \\
 0 & \frac{1}{4\xi^6} & 0 & \frac{1}{4} & 0 & \frac{\xi}{\sqrt{8}} \\
 \frac{\xi^5}{\sqrt{8}} & 0 & \frac{-1}{\sqrt{8}\xi} & 0 & \frac{1}{2} & 0 \\
 0 & \frac{-\xi^5}{\sqrt{8}} & 0 & \frac{1}{\sqrt{8}\xi} & 0 & \frac{1}{2}
\end{pmatrix},
\begin{pmatrix}
 \frac{1}{4} & 0 & \frac{1}{4 \xi^6} & 0 & \frac{\xi^5}{\sqrt{8}} & 0 \\
 0 & \frac{1}{4} & 0 & \frac{1}{4 \xi^6} & 0 & \frac{-\xi^5}{\sqrt{8}} \\
 \frac{\xi^6}{4} & 0 & \frac{1}{4} & 0 & \frac{-1}{\sqrt{8} \xi} & 0 \\
 0 & \frac{\xi^6}{4} & 0 & \frac{1}{4} & 0 & \frac{1}{\sqrt{8} \xi} \\
 \frac{1}{\sqrt{8} \xi^5} & 0 & \frac{-\xi}{\sqrt{8}} & 0 & \frac{1}{2} & 0 \\
 0 & \frac{-1}{\sqrt{8} \xi^5} & 0 & \frac{\xi}{\sqrt{8}} & 0 & \frac{1}{2}
\end{pmatrix}, \\
&
\begin{pmatrix}
 \frac{1}{4} & 0 & 0 & \frac{\xi^6}{4} & \frac{1}{4} & \frac{1}{4 \xi^6} \\
 0 & \frac{1}{4} & \frac{\xi^6}{4} & 0 & \frac{\xi^6}{4} & \frac{-1}{4} \\
 0 & \frac{1}{4\xi^6} & \frac{1}{4} & 0 & \frac{1}{4} & \frac{\xi^6}{4} \\
 \frac{1}{4\xi^6} & 0 & 0 & \frac{1}{4} & \frac{1}{4 \xi^6} & \frac{-1}{4} \\
 \frac{1}{4} & \frac{1}{4\xi^6} & \frac{1}{4} & \frac{\xi^6}{4} & \frac{1}{2} & 0 \\
 \frac{\xi^6}{4} & \frac{-1}{4} & \frac{1}{4\xi^6} & \frac{-1}{4} & 0 & \frac{1}{2}
\end{pmatrix},
\begin{pmatrix}
 \frac{1}{4} & 0 & 0 & \frac{1}{4} & \frac{(-1-\sqrt{3})}{8} & \frac{(1-\sqrt{3})}{8} \\
 0 & \frac{1}{4} & \frac{-1}{4} & 0 & \frac{(1-\sqrt{3})}{8} & \frac{(1+\sqrt{3})}{8} \\
 0 & \frac{-1}{4} & \frac{1}{4} & 0 & \frac{(\sqrt{3}-1)}{8} & \frac{(-1-\sqrt{3})}{8} \\
 \frac{1}{4} & 0 & 0 & \frac{1}{4} & \frac{(-1-\sqrt{3})}{8} & \frac{(1-\sqrt{3})}{8} \\
 \frac{-1-\sqrt{3}}{8} & \frac{1-\sqrt{3}}{8} & \frac{\sqrt{3}-1}{8} & \frac{-1-\sqrt{3}}{8} & \frac{1}{2} & 0 \\
 \frac{1-\sqrt{3}}{8} & \frac{1+\sqrt{3}}{8} & \frac{-1-\sqrt{3}}{8} & \frac{1-\sqrt{3}}{8} & 0 & \frac{1}{2}
\end{pmatrix},\\
&
\begin{pmatrix}
 \frac{1}{4} & 0 & 0 & \frac{1}{4\xi^6} & \frac{1}{4} & \frac{\xi^6}{4} \\
 0 & \frac{1}{4} & \frac{1}{4\xi^6} & 0 & \frac{1}{4 \xi^6} & \frac{-1}{4} \\
 0 & \frac{\xi^6}{4} & \frac{1}{4} & 0 & \frac{1}{4} & \frac{1}{4 \xi^6} \\
 \frac{\xi^6}{4} & 0 & 0 & \frac{1}{4} & \frac{\xi^6}{4} & \frac{-1}{4} \\
 \frac{1}{4} & \frac{\xi^6}{4} & \frac{1}{4} & \frac{1}{4\xi^6} & \frac{1}{2} & 0 \\
 \frac{1}{4\xi^6} & \frac{-1}{4} & \frac{\xi^6}{4} & \frac{-1}{4} & 0 & \frac{1}{2}
\end{pmatrix},
\begin{pmatrix}
 \frac{1}{4} & 0 & 0 & \frac{-1}{4} & \frac{(-1-\sqrt{3})}{8} & \frac{(\sqrt{3}-1)}{8} \\
 0 & \frac{1}{4} & \frac{1}{4} & 0 & \frac{(\sqrt{3}-1)}{8} & \frac{(1+\sqrt{3})}{8} \\
     0 & \frac{1}{4} & \frac{1}{4} & 0 & \frac{(\sqrt{3}-1)}{8} & \frac{(1+\sqrt{3})}{8} \\
 \frac{-1}{4} & 0 & 0 & \frac{1}{4} & \frac{(1+\sqrt{3})}{8} & \frac{(1-\sqrt{3})}{8} \\
 \frac{-1-\sqrt{3}}{8 } & \frac{\sqrt{3}-1}{8 } & \frac{\sqrt{3}-1}{8 } & \frac{1+\sqrt{3}}{8 } & \frac{1}{2} & 0 \\
 \frac{\sqrt{3}-1}{8 } & \frac{1+\sqrt{3}}{8 } & \frac{1+\sqrt{3}}{8 } & \frac{1-\sqrt{3}}{8 } & 0 & \frac{1}{2}
\end{pmatrix}.
\end{align*}
A direct yet tedious calculation shows that $X_1,\dots,X_9$ are projections, $\sum_jX_j=3I$ and $X_jX_kX_j=\frac14 X_j$ for $j\neq k$. Therefore $\pi(x_j)=X_j$ defines a representation $\pi:\cA_S\to \mtxc{6}$.
Furthermore, one can verify that the span of $\{X_{j_1}X_{j_2}\colon 1\le j_1,j_2\le 9\}$ has dimension 25. Hence $\pi$ is not a 2-fold inflation of a $3$-dimensional representation or a direct sum of two $3$-dimensional representations (since $25>9,18$), so $\pi$ is an irreducible representation by Proposition \ref{p:rk}.
\end{exa}

\end{appendices}

\printbibliography

\end{document}